%% file: 2hamDown.tex
\newif\iflncs
\lncstrue
\iflncs
\documentclass{llncs}
\else
\documentclass{article}[11]
\usepackage{fullpage}
\usepackage{amsthm}
\fi

\usepackage{amsmath}
\usepackage{amssymb}
\usepackage{algorithmic}
\usepackage{subfig}
\usepackage{wrapfig}
\usepackage{cite}
\usepackage{colortbl}
\usepackage{tikz}
\usepackage{ifpdf}

\usepackage{commands-tam}

\ifpdf

  \usepackage[pdftex]{epsfig}
  \usepackage[pdftex]{hyperref}

\else

    \usepackage[dvips]{epsfig}
    \newcommand{\href}[2]{#2}

\fi

\iflncs
\newtheorem{observation}[theorem]{Observation}
\else
\theoremstyle{definition}
\newtheorem{theorem}{Theorem}[section]
\newtheorem{lemma}[theorem]{Lemma}
\newtheorem{definition}[theorem]{Definition}
\newtheorem{observation}[theorem]{Observation}

\newtheorem{corollary}[theorem]{Corollary}
\newtheorem{claim}[theorem]{Claim}
\fi
\newcommand{\ta}{\tilde{\alpha}}
\newcommand{\tb}{\tilde{\beta}}
\newcommand{\tg}{\tilde{\gamma}}

\newif\ifabstract
\newif\iffull

\abstracttrue

\ifabstract
	\fullfalse
\else
	\fulltrue
\fi
\newcounter{section-preserve}
\newcounter{theorem-preserve}
\newcounter{lemma-preserve}
\newcounter{corollary-preserve}
\newcommand{\blank}[1]{}
\newtoks\magicAppendix
\magicAppendix={}
\newtoks\magictoks
\newif\iflater
\laterfalse
\ifabstract
\long\def\later#1{\magictoks={#1}%
  \edef\magictodo{\noexpand\magicAppendix={\the\magicAppendix \par
    \the\magictoks%
  }}
  \magictodo}
\long\def\both#1{\magictoks={#1}%
  \edef\magictodo{\noexpand\magicAppendix={\the\magicAppendix \par
    \noexpand\setcounter{theorem-preserve}{\noexpand\arabic{theorem}}%
    \noexpand\setcounter{lemma-preserve}{\noexpand\arabic{lemma}}%
    \noexpand\setcounter{corollary-preserve}{\noexpand\arabic{corollary}}%
    \noexpand\setcounter{theorem}{\arabic{theorem}}%
    \noexpand\setcounter{lemma}{\arabic{lemma}}%
    \noexpand\setcounter{corollary}{\arabic{corollary}}%
    \noexpand\setcounter{section-preserve}{\noexpand\arabic{section}}%
    \noexpand\setcounter{section}{\arabic{section}}%
	\noexpand\let\noexpand\oldsection=\noexpand\thesection
	\noexpand\def\noexpand\thesection{\thesection}
	\noexpand\let\noexpand\oldlabel=\noexpand\label
	\noexpand\let\noexpand\label=\noexpand\blank
    \the\magictoks%
    \noexpand\setcounter{theorem}{\noexpand\arabic{theorem-preserve}}%
    \noexpand\setcounter{lemma}{\noexpand\arabic{lemma-preserve}}%
    \noexpand\setcounter{corollary}{\noexpand\arabic{corollary-preserve}}%
    \noexpand\setcounter{section}{\noexpand\arabic{section-preserve}}%
	\noexpand\let\noexpand\thesection=\noexpand\oldsection
	\noexpand\let\noexpand\label=\noexpand\oldlabel
  }}
  \magictodo
  \the\magictoks}
\else
\long\def\later#1{#1}
\long\def\both#1{#1}
\fi
\long\def\magicappendix{
	\latertrue%
	\the\magicAppendix%
}

\title{The Simulation Powers and Limitations of Higher Temperature Hierarchical Self-Assembly Systems\thanks{Supported in part by National Science Foundation Grant CCF-1422152.}}

\author{
  Jacob Hendricks%
        \thanks{Dept. of Comp. Sci. and Comp. Eng., University of Arkansas,
      \protect\url{jhendric@uark.edu}
     }
\and
  Matthew J. Patitz%
    \thanks{Dept. of Comp. Sci. and Comp. Eng., University of Arkansas,
      \protect\url{patitz@uark.edu}
      }
\and
 Trent A. Rogers%
        \thanks{Dept. of Comp. Sci. and Comp. Eng., University of Arkansas,
      \protect\url{tar003@uark.edu}
      This author's research was supported by the National Science Foundation Graduate Research Fellowship Program under Grant No. DGE-1450079.}
      }

\iflncs
\institute{}
\fi

\date{}

\begin{document}

\maketitle

\begin{abstract}
\vspace{-10pt}
In this paper, we extend existing results about simulation and intrinsic universality in a model of tile-based self-assembly.  Namely, we work within the 2-Handed Assembly Model (2HAM), which is a model of self-assembly in which assemblies are formed by square tiles that are allowed to combine, using glues along their edges, individually or as pairs of arbitrarily large assemblies in a hierarchical manner, and we explore the abilities of these systems to simulate each other when the simulating systems have a higher ``temperature'' parameter, which is a system wide threshold dictating how many glue bonds must be formed between two assemblies to allow them to combine.  It has previously been shown that systems with lower temperatures cannot simulate arbitrary systems with higher temperatures, and also that systems at some higher temperatures can simulate those at particular lower temperatures, creating an infinite set of infinite hierarchies of 2HAM systems with strictly increasing simulation power within each hierarchy.  These previous results relied on two different definitions of simulation, one (\emph{strong} simulation) seemingly more restrictive than the other (\emph{standard} simulation), but which have previously not been proven to be distinct.  Here we prove distinctions between them by first fully characterizing the set of pairs of temperatures such that the high temperature systems are intrinsically universal for the lower temperature systems (i.e. one tile set at the higher temperature can simulate any at the lower) using strong simulation.  This includes the first impossibility result for simulation downward in temperature.  We then show that lower temperature systems which cannot be simulated by higher temperature systems using the strong definition, can in fact be simulated using the standard definition, proving the distinction between the types of simulation.
\end{abstract}

\input{intro}

\input{definitions}

\input{UniformMapping}
\input{StrongPossible}
\input{StrongImpossible}
\input{no-shift}
\input{SimLadders2}

\ifabstract
\later{

}
\fi
\iffull

\fi
\renewcommand\refname{\vspace{-45pt}\section*{References\vspace{-20pt}}}
\bibliographystyle{plain}%
{\bibliography{tam,experimental_refs,ca}}

\ifabstract
\newpage
\appendix

\begin{center}
	\Huge\bfseries
	Technical Appendix
\end{center}

\magicappendix
\fi

\end{document}

%% file: intro.tex
\section{Introduction}

In computational theory, a powerful and widely used tool for determining the relative powers of systems is simulation.  For instance, in order to prove the equivalence, in terms of computational power, of Turing machines and various abstract models such as tag systems, counter machines, cellular automata, and tile-based self-assembly, systems have been developed in each which demonstrate their abilities to simulate arbitrary Turing machines, and vice versa.  This has been used to prove that whatever can be computed by a system within one model can also be computed by a system in another.  Additionally, the notion of a universal Turing machine is based upon the fact that there exist Turing machines which can simulate others.

The methods of simulation which are typically employed involve mappings of behaviors and states in one model or system to those in another, often following some ``natural'' mapping function, and also often in such a way that the simulation is guaranteed to generate the same final result as the simulated system, and maybe even some or all of its intermediate states.  Nonetheless, there is usually no requirement that the simulator ``do it the same way,'' i.e. the dynamical behavior of the simulator need not mirror that of the simulated. For instance, as one Turing machine $A$ simulates another, $B$, its head movements may be in a significantly different pattern than $B$'s since, for instance, it may frequently move to a special portion of the tape which encodes $B$'s transition table, then back to the ``data'' section.

While such types of simulation can be informative when asking questions about the equivalence of computational powers of systems, oftentimes it is the behavior of a system which is of interest, not just its ``output.''  Self-assembling systems, which are those composed of large numbers of relatively simple components which autonomously combine to form structures using only local interactions, often fall into this category since the actual ways in which they evolve and build structures are of key importance.  In this paper, we focus our attention on tile-based self-assembling systems in a model known as the 2-Handed Assembly Model (2HAM) \cite{Versus}, which is a generalization of the abstract Tile Assembly Model (aTAM) \cite{Winf98} in which the basic components are square \emph{tiles} which are able to bind to each other when they possess matching \emph{glues} on their edges.  In the aTAM, assembly occurs as tiles autonomously combine, with one tile at a time attaching to a growing assembly.  In the 2HAM, similar growth can occur, but it is possible for pairs of arbitrarily large assemblies (a.k.a. supertiles) to combine as well.
Because the dynamical behaviors of these systems are of such importance, work in these models (e.g. \cite{IUSA,2HAMIU,IUNeedsCoop,EquivCAandTAM,Signals3D,WoodsIU2013}) has turned to a notion of simulation developed within the domain of cellular automata, whose dynamical behaviors are also often of central importance.  This notion of simulation, called \emph{intrinsic universality} (see %
\cite{mazoyer1998inducing,DurandRoka,bulkingI,Delorme2011,Goles-etal-2011,Ollinger-CSP08,ollingerRichard2011four,arrighi2012intrinsic,arrighi2012intrinsically}
for some examples related to various models such as cellular automata), is defined in such a way that the simulations performed are essentially ``in place'' simulations which mirror the dynamics of the simulated systems, modulo a scale factor allowed the simulator.  Intrinsic universality has been used to show the existence of ``universal'' systems, somewhat analogous to universal Turing machines, which can simulate all other systems within a given model or class of systems, but in a dynamics-preserving way.  Previous work \cite{IUSA} has shown that there exists a single aTAM tile set which is capable of simulating any arbitrary aTAM system, and thus that tile set is intrinsically universal (IU) for the aTAM (and we also say that the aTAM is IU).  Further work in \cite{2HAMIU} showed that the 2HAM is much more complicated in terms of IU, with there existing hierarchies of 2HAM systems with strictly-increasing power of simulation.  These simulations are performed by scaled blocks of tiles known as macrotiles in the simulator used to simulate individual tiles in the simulated systems.  The simulation hierarchy in the 2HAM is based on a classification of systems separated by a system parameter known as the \emph{temperature}, which is the global threshold that specifies the minimum strength of glue bindings required for pairs of tiles or supertiles to combine.  It was proven in \cite{2HAMIU} that for every temperature $\tau \ge 2$, there exists a system at temperature $\tau$ such that no system at temperature $\tau' < \tau$ can simulate it.  However, they also showed that for each $\tau \ge 2$, the class of 2HAM systems at $\tau$ is IU.

The motivation of the current paper is to extend and further develop the results of \cite{2HAMIU}, especially Theorem 4 which states:  "There exists an infinite number of infinite hierarchies of 2HAM systems with strictly-increasing power (and temperature) that can simulate downward within their own hierarchy.''  Our results elucidate more details about this hierarchy, including proving important differences between different notions of simulation used to characterize intrinsically universal systems.  More specifically, different definitions of simulation have been used even within the IU results of \cite{2HAMIU}, with one referred to as \emph{strong} simulation and one as (\emph{standard}) simulation.  Strong simulation is a stricter notion essentially stating that whenever two supertiles in the simulated system $\calT$ can combine, every pair of macrotiles that represents them in the simulator $\calS$ must be able to (eventually) combine.  However, standard simulation simply requires that for each half of such a pair in the simulator, there must exist some mate with which it can eventually combine.  While both notions of simulation were utilized in \cite{2HAMIU}, no concrete distinction was proven in terms of what is or isn't possible between them.  Here, we first prove that higher temperature systems can strongly simulate lower temperature systems if and only if there is a relationship between the temperature values which we call a \emph{uniform mapping}.  We show that it is easy to find whether such a mapping exists between two temperatures and, if so, what one is, and prove that for each pair of temperatures $2 < \tau < \tau'$ where a uniform mapping exists from $\tau$ to $\tau'$, that there exists a tile set which, at temperature $\tau'$, is IU for the class of 2HAM systems at $\tau$.  We then prove that if no uniform mapping exists from $\tau$ to $\tau'$, then there exist systems at $\tau$ which cannot be strongly simulated by any system at $\tau'$, which is the first impossibility result for simulating downward in temperature that we are aware of, and is of interest because a natural intuition is that higher temperature systems are strictly more powerful. (However, we also show that for any given $\tau$ there are only a finite number of $\tau' > \tau$ to which a uniform mapping does not exist.)  Finally, we show that some systems which cannot be strongly simulated by higher temperature systems when no uniform mapping exists between temperatures can in fact be simulated from the higher temperature using the standard definition of simulation.  This shows the first clear distinction between what is possible under the various definitions, and that the notion of strong simulation is provably more restrictive than that of (standard) simulation since the set of systems which can be simulated by a higher temperature system is strictly greater than that which can be strongly simulated.

In the next section we provide the definitions of the model and framework for our results, then provide an overview of our results in the following sections.  Please note that due to space constraints, proofs have been placed in the Appendix.

%% file: definitions.tex
\section{Definitions}\label{sec:definitions}
\subsection{Informal definition of the 2HAM}

Here we give a brief, informal, sketch of the 2HAM.  Please see Section~\ref{sec:2ham-formal} for a more formal definition.
The 2HAM \cite{AGKS05g,DDFIRSS07} is a generalization of the aTAM \cite{Winf98}, and in both the basic components are ``tiles''.
A \emph{tile type} is a unit square with four sides, each having a \emph{glue} consisting of a \emph{label} (a finite string) and \emph{strength} (a non-negative integer).   We assume a finite set $T$ of tile types, but an infinite number of copies of each tile type, each copy referred to as a \emph{tile}.
A \emph{supertile} is (the set of all translations of) a positioning of tiles on the integer lattice $\Z^2$.  Two adjacent tiles in a supertile \emph{interact} if the glues on their abutting sides are equal and have positive strength.
Each supertile induces a \emph{binding graph}, a grid graph whose vertices are tiles, with an edge between two tiles if they interact.
The supertile is \emph{$\tau$-stable} if every cut of its binding graph has strength at least $\tau$, where the weight of an edge is the strength of the glue it represents.
That is, the supertile is stable if at least energy $\tau$ is required to separate the supertile into two parts.
A 2HAM \emph{tile assembly system} (TAS) is a triple $\calT = (T,S,\tau)$, where $T$ is a finite tile set, $S$ is a set of \emph{seed} supertiles over $T$, and $\tau$ is the \emph{temperature}, usually 1 or 2.  When $S$ is solely an infinite number of each of the singleton tiles of $T$, we call that the default initial state, and for shorthand notion refer to a TAS with a default initial state simply as a pair $\calT = (T,\tau)$.
Given a TAS $\calT=(T,S,\tau)$, a supertile is \emph{producible}, written as $\alpha \in \prodasm{T}$ if either it is a (super)tile in $S$, or it is the $\tau$-stable result of translating two producible assemblies without overlap.  That is, any $\tau$-stable supertile which can result from some positioning of two producible supertiles, so that they do not overlap and they bind with at least strength $\tau$, is itself a producible supertile. This potentially allows for the combination of pairs of arbitrary large supertiles.
A supertile $\alpha$ is \emph{terminal}, written as $\alpha \in \termasm{T}$ if for every producible supertile $\beta$, $\alpha$ and $\beta$ cannot be $\tau$-stably attached.

\ifabstract
\later{

\section{Formal definition of the 2HAM}\label{sec:2ham-formal}

We now formally define the 2HAM.

Two assemblies $\alpha$ and $\beta$ are \emph{disjoint} if $\dom \alpha \cap \dom \beta = \emptyset.$
For two assemblies $\alpha$ and $\beta$, define the \emph{union} $\alpha \cup \beta$ to be the assembly defined for all $\vec{x}\in\Z^2$ by $(\alpha \cup \beta)(\vec{x}) = \alpha(\vec{x})$ if $\alpha(\vec{x})$ is defined, and $(\alpha \cup \beta)(\vec{x}) = \beta(\vec{x})$ otherwise. Say that this union is \emph{disjoint} if $\alpha$ and $\beta$ are disjoint.

The \emph{binding graph of} an assembly $\alpha$ is the grid graph
$G_\alpha = (V, E )$, where $V =
\dom{\alpha}$, and $\{\vec{m}, \vec{n}\} \in E$ if and only if (1)
$\vec{m} - \vec{n} \in U_2$, (2)
$\lab_{\alpha(\vec{m})}\left(\vec{n} - \vec{m}\right) =
\lab_{\alpha(\vec{n})}\left(\vec{m} - \vec{n}\right)$, and (3)
$\strength_{\alpha(\vec{m})}\left(\vec{n} -\vec{m}\right) > 0$.
Given $\tau \in \mathbb{N}$, an
assembly is $\tau$-\emph{stable} (or simply \emph{stable} if $\tau$ is understood from context), if it
cannot be broken up into smaller assemblies without breaking bonds
of total strength at least $\tau$; i.e., if every cut of $G_\alpha$
has weight at least $\tau$, where the weight of an edge is the strength of the glue it represents. In contrast to the model of Wang tiling, the nonnegativity of the strength function implies that glue mismatches between adjacent tiles do not prevent a tile from binding to an assembly, so long as sufficient binding strength is received from the (other) sides of the tile at which the glues match.

For assemblies $\alpha,\beta:\Z^2 \dashrightarrow T$ and $\vec{u} \in \Z^2$, we write $\alpha+\vec{u}$ to denote the assembly defined for all $\vec{x}\in\Z^2$ by $(\alpha+\vec{u})(\vec{x}) = \alpha(\vec{x}-\vec{u})$, and write $\alpha \simeq \beta$ if there exists $\vec{u}$ such that $\alpha + \vec{u} = \beta$; i.e., if $\alpha$ is a translation of $\beta$. Given two assemblies $\alpha,\beta:\Z^2 \dashrightarrow T$, we say $\alpha$ is a \emph{subassembly} of $\beta$, and we write $\alpha \sqsubseteq \beta$, if $S_\alpha \subseteq S_\beta$ and, for all points $p \in S_\alpha$, $\alpha(p) = \beta(p)$.
Define the \emph{supertile} of $\alpha$ to be the set $\ta = \setr{\beta}{\alpha \simeq \beta}$.
A supertile $\ta$ is \emph{$\tau$-stable} (or simply \emph{stable}) if all of the assemblies it contains are $\tau$-stable; equivalently, $\ta$ is stable if it contains a stable assembly, since translation preserves the property of stability. Note also that the notation $|\ta| \equiv |\alpha|$ is the size of the supertile (i.e., number of tiles in the supertile) is well-defined, since translation preserves cardinality (and note in particular that even though we define $\ta$ as a set, $|\ta|$ does not denote the cardinality of this set, which is always $\aleph_0$).

For two supertiles $\ta$ and $\tb$, and temperature $\tau\in\N$, define the \emph{combination} set $C^\tau_{\ta,\tb}$ to be the set of all supertiles $\tg$ such that there exist $\alpha \in \ta$ and $\beta \in \tb$ such that (1) $\alpha$ and $\beta$ are disjoint (steric protection), (2) $\gamma \equiv \alpha \cup \beta$ is $\tau$-stable, and (3) $\gamma \in \tg$. That is, $C^\tau_{\ta,\tb}$ is the set of all $\tau$-stable supertiles that can be obtained by ``attaching'' $\ta$ to $\tb$ stably, with $|C^\tau_{\ta,\tb}| > 1$ if there is more than one position at which $\beta$ could attach stably to $\alpha$.

It is common with seeded assembly to stipulate an infinite number of copies of each tile, but our definition allows for a finite number of tiles as well. Our definition also allows for the growth of infinite assemblies and finite assemblies to be captured by a single definition, similar to the definitions of \cite{jSSADST} for seeded assembly.

Given a set of tiles $T$, define a \emph{state} $S$ of $T$ to be a multiset of supertiles, or equivalently, $S$ is a function mapping supertiles of $T$ to $\N \cup \{\infty\}$, indicating the multiplicity of each supertile in the state. We therefore write $\ta \in S$ if and only if $S(\ta) > 0$.

A \emph{(two-handed) tile assembly system} (\emph{TAS}) is an ordered triple $\mathcal{T} = (T, S, \tau)$, where $T$ is a finite set of tile types, $S$ is the \emph{initial state}, and $\tau\in\N$ is the temperature. If not stated otherwise, assume that the initial state $S$ is defined $S(\ta) = \infty$ for all supertiles $\ta$ such that $|\ta|=1$, and $S(\tb) = 0$ for all other supertiles $\tb$. That is, $S$ is the state consisting of a countably infinite number of copies of each individual tile type from $T$, and no other supertiles. In such a case we write $\calT = (T,\tau)$ to indicate that $\calT$ uses the default initial state.  For notational convenience we sometimes describe $S$ as a set of supertiles, in which case we actually mean that  $S$ is a multiset of supertiles with infinite count of each supertile. We also assume that, in general, unless stated otherwise, the count for any supertile in the initial state is infinite.

Given a TAS $\calT=(T,S,\tau)$, define an \emph{assembly sequence} of $\calT$ to be a sequence of states $\vec{S} = (S_i \mid 0 \leq i < k)$ (where $k = \infty$ if $\vec{S}$ is an infinite assembly sequence), and $S_{i+1}$ is constrained based on $S_i$ in the following way: There exist supertiles $\ta,\tb,\tg$ such that (1) $\tg \in C^\tau_{\ta,\tb}$, (2) $S_{i+1}(\tg) = S_{i}(\tg) + 1$,\footnote{with the convention that $\infty = \infty + 1 = \infty - 1$} (3) if $\ta \neq \tb$, then $S_{i+1}(\ta) = S_{i}(\ta) - 1$, $S_{i+1}(\tb) = S_{i}(\tb) - 1$, otherwise if $\ta = \tb$, then $S_{i+1}(\ta) = S_{i}(\ta) - 2$, and (4) $S_{i+1}(\tilde{\omega}) = S_{i}(\tilde{\omega})$ for all $\tilde{\omega} \not\in \{\ta,\tb,\tg\}$.
That is, $S_{i+1}$ is obtained from $S_i$ by picking two supertiles from $S_i$ that can attach to each other, and attaching them, thereby decreasing the count of the two reactant supertiles and increasing the count of the product supertile. If $S_0 = S$, we say that $\vec{S}$ is \emph{nascent}.

Given an assembly sequence $\vec{S} = (S_i \mid 0 \leq i < k)$ of $\calT=(T,S,\tau)$ and a supertile $\tg \in S_i$ for some $i$, define the \emph{predecessors} of $\tg$ in $\vec{S}$ to be the multiset $\pred_{\vec{S}}(\tg) = \{\ta,\tb\}$ if $\ta,\tb \in S_{i-1}$ and $\ta$ and $\tb$ attached to create $\tg$ at step $i$ of the assembly sequence, and define $\pred_{\vec{S}}(\tg) = \{ \tg \}$ otherwise. Define the \emph{successor} of $\tg$ in $\vec{S}$ to be $\succ_{\vec{S}}(\tg)=\ta$ if $\tg$ is one of the predecessors of $\ta$ in $\vec{S}$, and define $\succ_{\vec{S}}(\tg)=\tg$ otherwise. A sequence of supertiles $\vec{\ta} = (\ta_i \mid 0 \leq i < k)$ is a \emph{supertile assembly sequence} of $\calT$ if there is an assembly sequence $\vec{S} = (S_i \mid 0 \leq i < k)$ of $\calT$ such that, for all $1 \leq i < k$, $\succ_{\vec{S}}(\ta_{i-1}) = \ta_i$, and $\vec{\ta}$ is \emph{nascent} if $\vec{S}$ is nascent.

The \emph{result} of a supertile assembly sequence $\vec{\ta}$ is the unique supertile $\res{\vec{\ta}}$ such that there exist an assembly $\alpha \in \res{\vec{\ta}}$ and, for each $0 \leq i < k$, assemblies $\alpha_i \in \ta_i$ such that $\dom{\alpha} = \bigcup_{0 \leq i < k}{\dom{\alpha_i}}$ and, for each $0 \leq i < k$, $\alpha_i \sqsubseteq \alpha$.  For all supertiles $\ta,\tb$, we write $\ta \to_\calT \tb$ (or $\ta \to \tb$ when $\calT$ is clear from context) to denote that there is a supertile assembly sequence $\vec{\ta} = ( \ta_i \mid 0 \leq i < k )$ such that $\ta_0 = \ta$ and $\res{\vec{\ta}} = \tb$. It can be shown using the techniques of \cite{Roth01} for seeded systems that for all two-handed tile assembly systems $\calT$ supplying an infinite number of each tile type, $\to_\calT$ is a transitive, reflexive relation on supertiles of $\calT$. We write $\ta \to_\calT^1 \tb$ ($\ta \to^1 \tb$) to denote an assembly sequence of length 1 from $\ta$ to $\tb$ and $\ta \to_\calT^{\leq 1} \tb$ ($\ta \to^{\leq 1} \tb$) to denote an assembly sequence of length 1 from $\ta$ to $\tb$ if $\ta \ne \tb$ and an assembly sequence of length 0 otherwise.

A supertile $\ta$ is \emph{producible}, and we write $\ta \in \prodasm{\calT}$, if it is the result of a nascent supertile assembly sequence. A supertile $\ta$ is \emph{terminal} if, for all producible supertiles $\tb$, $C^\tau_{\ta,\tb} = \emptyset$.\footnote{Note that a supertile $\ta$ could be non-terminal in the sense that there is a producible supertile $\tb$ such that $C^\tau_{\ta,\tb} \neq \emptyset$, yet it may not be possible to produce $\ta$ and $\tb$ simultaneously if some tile types are given finite initial counts, implying that $\ta$ cannot be ``grown'' despite being non-terminal. If the count of each tile type in the initial state is $\infty$, then all producible supertiles are producible from any state, and the concept of terminal becomes synonymous with ``not able to grow'', since it would always be possible to use the abundant supply of tiles to assemble $\tb$ alongside $\ta$ and then attach them.} Define $\termasm{\calT} \subseteq \prodasm{\calT}$ to be the set of terminal and producible supertiles of $\calT$. $\calT$ is \emph{directed} (a.k.a., \emph{deterministic}, \emph{confluent}) if $|\termasm{\calT}| = 1$.

} %
\fi

\subsection{Definitions for simulation}\label{sec:defsSim}
In this subsection, we formally define what it means for one 2HAM TAS to ``simulate'' another 2HAM TAS. The definitions presented in this (and the next) subsection are based on the simulation definitions from \cite{Versus, IUSA, IUNeedsCoop} and are included here for the sake of completeness.  We will be describing how the assembly process followed by a system $\mathcal{T}$ is simulated by a system $\mathcal{U}$, which we will call the \emph{simulator}.  The simulation performed by $\mathcal{U}$ will be such that the assembly process followed by $\mathcal{U}$ mirrors that of the simulated system $\mathcal{T}$, but with the individual tiles of $\mathcal{T}$ represented by (potentially large) square blocks of tiles in $\mathcal{U}$ called \emph{macrotiles}.  We now provide the definitions necessary to define $\mathcal{U}$ as a valid simulator of $\mathcal{T}$.
For a tileset $T$, let $A^T$ and $\tilde{A}^T$ denote the set of all assemblies over $T$ and all supertiles over $T$ respectively. Let $A^T_{< \infty}$ and $\tilde{A}^T_{< \infty}$ denote the set of all finite assemblies over $T$ and all finite supertiles over $T$ respectively.

In what follows, let $U$ be a tile set. An $m$-\emph{block assembly}, or {\em macrotile},  over tile set $U$ is a partial function $\gamma : \mathbb{Z}_m \times \mathbb{Z}_m \dashrightarrow U$, where $\mathbb{Z}_m = \{ 0,1,\ldots m-1 \}$.  Let $B^U_m$ be the set of all $m$-block assemblies over $U$. The $m$-block with no domain is said to be $\emph{empty}$.  For an arbitrary assembly $\alpha \in A^U$ define $\alpha^m_{x,y}$ to be the $m$-block defined by $\alpha^m_{x,y}(i,j) = \alpha(mx+i,my+j)$ for $0\leq i,j < m$.

For a partial function $R: B^{U}_m \dashrightarrow T$, define the \emph{assembly representation function} $R^*: A^{U} \dashrightarrow A^T$ such that $R^*(\alpha) = \beta$ if and only if $\beta(x,y) = R(\alpha^m_{x,y})$ for all $x,y \in \mathbb{Z}^2$.
    Further,
     $\alpha$ is said to map \emph{cleanly} to $\beta$ under $R^*$ if either (1) for all non empty blocks $\alpha^m_{x,y}$, $(x+u,y+v) \in \dom{\beta}$ for some $u,v \in \{-1,0,1\}$ such that $u^2+v^2 < 2$, or (2) $\alpha$ has at most one non-empty $m$-block $\alpha^m_{x,y}$. In other words, we allow for the existence of simulator ``fuzz'' directly north, south, east or west of a simulator  macrotile, but we exclude the possibility of diagonal fuzz.

For a given \emph{assembly representation function} $R^*$, define the \emph{supertile representation function} $\tilde{R}: \tilde{A}^{U} \dashrightarrow \mathcal{P}(A^T)$ such that $\tilde{R}(\ta) = \{R^*(\alpha) | \alpha \in \ta \}$. $\ta$ is said to \emph{map cleanly} to $\tilde{R}(\ta)$ if $\tilde{R}(\ta)\in \tilde{A}^T$ and $\alpha$ maps cleanly to $R^*(\alpha)$ for all~$\alpha \in \ta$.
In the following definitions, let $\mathcal{T} = \left(T,S,\tau\right)$  be a 2HAM TAS and, for some initial configuration $S_{\mathcal{T}}$, that depends on $\mathcal{T}$, let $\mathcal{U} = \left(U,S_{\mathcal{T}},\tau'\right)$ be a 2HAM TAS, and let $R$ be an $m$-block representation function $R: B^U_m \dashrightarrow T$.

\begin{definition}\label{scott-defn:alt-equiv-prod}
We say that $\mathcal{U}$ and $\mathcal{T}$ have \emph{equivalent productions} (at scale factor $m$), and we write $\mathcal{U} \Leftrightarrow_R \mathcal{T}$ if the following conditions hold:
\begin{enumerate}
    \item \label{scott-defn:simulate:equiv_prod_a}$\left\{\tilde{R}(\ta) | \ta \in \prodasm{\mathcal{U}}\right\} = \prodasm{\mathcal{T}}$.
    \item \label{scott-defn:simulate:equiv_prod_a2}$\left\{\tilde{R}(\ta) | \ta \in \termasm{\mathcal{U}}\right\} = \termasm{\mathcal{T}}$.    \item \label{scott-defn:simulate:equiv_prod_b}For all $\ta \in \prodasm{\mathcal{U}}$, $\ta$ maps cleanly to $\tilde{R}(\ta)$
\end{enumerate}
\end{definition}

Equivalent production tells us that a simulating system $\mathcal{U}$ produces exactly the same set of assemblies as the simulated system $\mathcal{T}$, modulo scale factor (with the representation function providing the mapping of assemblies between the systems).  While this is a powerful set of conditions ensuring that the simulator makes the same assemblies, it does not provide a guarantee that the simulator makes them in the \emph{same way}.  Namely, we desire a simulator to make the same assemblies, but also by following the same assembly sequences (again modulo scale and application of the representation function).  We call this the \emph{dynamics} of the systems and capture the necessary equivalence in the next few definitions.  It is notable that the conditions required for the dynamics of the systems to be equivalent, \emph{following} and \emph{modeling}, are strong enough that equivalent production follows in a straightforward way from them, and therefore is redundant.  However, we include it for completeness and clarity.\begin{definition}\label{scott-defn:alt-equiv-dynamic-t-to-s}
We say that $\mathcal{T}$ \emph{follows} $\mathcal{U}$ (at scale factor $m$), and we write $\mathcal{T} \dashv_R \mathcal{U}$ if, for any $\ta, \tb \in \prodasm{\mathcal{U}}$ such that $\ta \rightarrow_{\mathcal{U}}^1 \tb$, $\tilde{R}(\ta) \rightarrow_\mathcal{T}^{\leq 1} \tilde{R}\left(\tb\right)$.
\end{definition}

\begin{definition}\label{scott-defn:alt-equiv-dyanmic-s-to-t-weak}
We say that $\mathcal{U}$ \emph{weakly models} $\mathcal{T}$ (at scale factor $m$), and we write $\mathcal{U} \models^-_R \mathcal{T}$ if, for any $\ta, \tb \in \prodasm{\mathcal{T}}$ such that $\ta \rightarrow_\mathcal{T}^1 \tb$, for all $\ta' \in \prodasm{\mathcal{U}}$ such that $\tilde{R}(\ta')=\ta$, there exists an $\ta'' \in \prodasm{\mathcal{U}}$ such that $\tilde{R}(\ta'')=\ta$, $\ta' \rightarrow_{\mathcal{U}} \ta''$, and $\ta'' \rightarrow_{\mathcal{U}}^1 \tb'$ for some $\tb' \in \prodasm{\mathcal{U}}$ with $\tilde{R}\left(\tb'\right)=\tb$.
\end{definition}
\begin{definition}\label{scott-defn:alt-equiv-dyanmic-s-to-t-strong}
We say that $\mathcal{U}$ \emph{strongly models} $\mathcal{T}$ (at scale factor $m$), and we write $\mathcal{U} \models^+_R \mathcal{T}$ if for any $\ta$, $\tb \in \prodasm{\mathcal{T}}$ such that $\tilde{\gamma} \in C^{\tau}_{\ta , \tb}$, then for all $\ta', \tb' \in \prodasm{\mathcal{U}}$ such that $\tilde{R}(\ta')=\ta$ and $\tilde{R}\left(\tb'\right)=\tb$, it must be that there exist $\ta'', \tb'', \tilde{\gamma}' \in \prodasm{\mathcal{U}}$, such that $\ta' \rightarrow_{\mathcal{U}} \ta''$, $\tb' \rightarrow_{\mathcal{U}} \tb''$, $\tilde{R}(\ta'')=\ta$, $\tilde{R}\left(\tb''\right)=\tb$, $\tilde{R}(\tilde{\gamma}')=\tilde{\gamma}$, and $\tilde{\gamma}' \in C^{\tau'}_{\ta'', \tb''}$.
\end{definition}

\begin{definition}\label{scott-defn:alt-simulate}
Let $\mathcal{U} \Leftrightarrow_R \mathcal{T}$ and $\mathcal{T} \dashv_R \mathcal{U}$.
\begin{enumerate}
    \item \label{scott-defn:alt-weak-simulate} $\mathcal{U}$ \emph{simulates} $\mathcal{T}$ (at scale factor $m$) if $\mathcal{U} \models^-_R \mathcal{T}$.
    \item \label{scott-defn:alt-strong-simulate} $\mathcal{U}$ \emph{strongly simulates} $\mathcal{T}$ (at scale factor $m$) if $\mathcal{U} \models_R^+ \mathcal{T}$.
\end{enumerate}
\end{definition}
For simulation, we require that when a simulated supertile $\ta$ may grow, via one combination attachment, into a second supertile $\tb$, then any simulator supertile that maps to $\ta$ must also grow into a simulator supertile that maps to $\tb$. The converse should also be true.  For strong simulation, in addition to requiring that all supertiles mapping to $\ta$ must be capable of growing into a supertile mapping to $\tb$ when $\ta$ can grow into $\tb$ in the simulated system, we further require that this growth can take place by the attachment of $\emph{any}$ supertile mapping to $\tilde{\gamma}$, where $\tilde{\gamma}$ is the supertile that attaches to $\ta$ to get $\tb$.

Note that, by these definitions, strong simulation implies simulation.  That is, if system $\calT_1$ strongly simulates $\calT_2$ then it also simulates $\calT_2$.

\subsection{Intrinsic universality}
\newcommand{\REPL}{\mathsf{REPR}}
\newcommand{\frakC}{\mathfrak{C}}

Let $\REPL$ denote the set of all $m$-block (or macrotile) representation functions.
Let $\frakC$ be a class of tile assembly systems, and let $U$ be a tile set. %
We say $U$ is \emph{intrinsically universal} for $\frakC$ if there are computable functions $\mathcal{R}:\frakC \to \REPL$ and $\mathcal{S}:\frakC \to \left(A^U_{< \infty} \rightarrow \mathbb{N} \cup \{\infty\}\right)$, and a $\tau'\in\Z^+$ such that, for each $\mathcal{T} = (T,S,\tau) \in \frakC$, there is a constant $m\in\N$ such that, letting $R = \mathcal{R}(\mathcal{T})$, $S_\mathcal{T}=\mathcal{S}(\mathcal{T})$, and $\mathcal{U}_\mathcal{T} = (U,S_\mathcal{T},\tau')$, $\mathcal{U}_\mathcal{T}$ simulates $\mathcal{T}$ at scale $m$ and using macrotile representation function $R$.
That is, $\mathcal{R}(\mathcal{T})$ gives a representation function $R$ that interprets macrotiles (or $m$-blocks) of $\mathcal{U}_\mathcal{T}$ as assemblies of $\mathcal{T}$, and $\mathcal{S}(\mathcal{T})$ gives the initial state used to create the necessary macrotiles from $U$ to represent $\mathcal{T}$ subject to the constraint that no macrotile in $S_{\calT}$ can be larger than a single $m \times m$ square.

%% file: UniformMapping.tex
\section{Uniform Mappings}\label{sec:uniform-mappings}

In this section, we define \emph{uniform mapping} and \emph{almost linear} uniform mapping, which will provide the basis for our results related to strong simulation.  We then prove a set of facts about pairs of temperatures and these mappings, most notably that it is ``easy'' to find a uniform mapping between temperatures if one exists.

\ifabstract
\later{
\section{Proofs from Section~\ref{sec:uniform-mappings}: Uniform Mappings}
}%
\fi

\begin{definition}
Let $E = \{n | n \in \mathbb{N} \textmd{ and } n \leq Q\}$ and $F = \{n | n \in \mathbb{N} \textmd{ and } n \leq R\}$ for some $Q, R \in \mathbb{Z}^+$ with $Q \leq R$.  Let $S$ be a multiset consisting of members from $E$.  Then we say that there is a \emph{uniform mapping} $M$ from $E$ to $F$ if there exists a function $M:E \rightarrow F$ such that $\sum\limits_{x \in S} M(x) \geq R$ if and only if $\sum\limits_{x \in S} x \geq Q$.
\end{definition}

We say that there is a uniform mapping from $\tau$ to $\tau'$ provided that there exists a uniform mapping from $\{1, 2, ..., \tau\}$ to $\{1, 2, ..., \tau'\}$.

\begin{definition}
Let $E = \{n | n \in \mathbb{N} \textmd{ and } n \leq Q\}$ and $F = \{n | n \in \mathbb{N} \textmd{ and } n \leq R\}$ for some $Q, R \in \mathbb{Z}^+$ with $Q \leq R$, and let $M:E \rightarrow F$ be a uniform mapping from $E$ to $F$.  We say that $M$ is \emph{almost linear} if there exists a $c \in \mathbb{N}$ such that for all $e \in (E - \{Q\})$, $M(e) = ce$, and $M(Q) = R$.
\end{definition}

If a uniform mapping is almost linear, that means that other than for the greatest value in the domain of the mapping, the mapping of a number $x$ is simply $x$ times some constant $c$, where $c$ is constant for the mapping.

\both{
\begin{lemma}\label{lem:almost-linear}
There exists a uniform mapping from $E = \{1, ..., \tau\}$ to $F=\{1, ..., \tau'\}$ if and only if there exists an almost linear uniform mapping from $E$ to $F$.
\end{lemma}
}%

\ifabstract
\later{

\begin{proof}
For the first direction, suppose that there exists a uniform mapping $M'$ from $E$ to $F$.  This implies that $M'$ is such that $\tau' \leq \tau M'(1)$ and $\tau' > (\tau-1)M'(1)$.  We can now construct a mapping $M_1:E \rightarrow \mathbb{N}$ in the following manner.  For all $x \in E$, set $M_1(x) = M'(1)x$.  Let $S$ be an arbitrary multiset consisting of members from $E$.  We claim that $M_1$ is such that $\sum\limits_{x \in S} M_1(x) \geq \tau'$ if and only if $\sum\limits_{x \in S} x \geq \tau$.

Suppose $S$ is such that $\sum\limits_{x \in S} x \geq \tau$.  To see that $\sum\limits_{x \in S} M_1(x) \geq \tau'$, observe that
\begin{eqnarray*}
\sum\limits_{x \in S} M_1(x) &=& \sum\limits_{x \in S} xM'(1) \\
                             &=& M'(1) \sum\limits_{x \in S} x \\
                             &\geq& M'(1) \tau \\
                             &\geq& \tau'. \\
\end{eqnarray*}

Now, assume $S$ is such that $\sum\limits_{x \in S} M_1(x) \geq \tau'$.  Notice that this implies
\begin{eqnarray*}
\sum\limits_{x \in S} x &=& \sum\limits_{x \in S} \frac{M_1(x)}{M'(1)} \\
                        &=& \frac{1}{M'(1)}\sum\limits_{x \in S} M_1(x) \\
                        &\geq& \frac{1}{M'(1)} \tau' \\
                        &>& (\tau - 1) \\
                        &\geq& \tau. \\
\end{eqnarray*}
The second to last inequality comes from a simple rearrangement of the above observation $\tau' > (\tau-1)M'(1)$.  %

We can now construct an almost linear uniform mapping $M: E \rightarrow F$ defined by
\begin{displaymath}
   M(x) = \left\{
     \begin{array}{lr}
       M'(1)x & : x \neq \tau\\
       \tau'  & : x = \tau.
     \end{array}
   \right.
\end{displaymath}

This map is clearly almost linear, since the $c$ of the definition of almost linear is $M'(1)$ here, and $P = \tau$ and $Q = \tau'$.  Additionally, the range of $M$ is $F$ because $\tau \le \tau'$ and, other than $M(\tau) = \tau'$, the maximum value of $M$ occurs at $M(\tau-1) = M'(1)(\tau-1) < \tau'$.

The other direction of the proof follows directly from the fact that an almost linear uniform mapping is a uniform mapping.
\qed
\end{proof}

}%
\fi

\both{
\begin{corollary}\label{cor:mapping-constant}
For $\tau, \tau' \in \mathbb{Z}^+$ where $\tau \le \tau'$, a uniform mapping from $\tau$ to $\tau'$ exists if and only if there exists a constant $c \in \mathbb{N}$ such that $c(\tau-1) < \tau' \le c\tau$.
\end{corollary}
}%

\ifabstract
\later{

\begin{proof}
The proof of Corollary~\ref{cor:mapping-constant} is a direct result of Lemma~\ref{lem:almost-linear} and its proof. First, by that lemma we know that a uniform mapping from $\tau$ to $\tau'$ exists if and only if an almost linear uniform mapping exists from $\tau$ to $\tau'$. Second, by the proof of that lemma we see that an almost linear uniform mapping exists if there is some value $M'(1)$ such that $\tau' \leq \tau M'(1)$ and $\tau' > (\tau-1)M'(1)$, and since $M'(1) \in \mathbb{N}$, we can simply define $c = M'(1)$ and thus see that $c(\tau-1) < \tau' \le c\tau$.
\qed
\end{proof}

}%
\fi

\both{
\begin{corollary}\label{cor:no-uniform-mapping}
Let $\tau \in \mathbb{Z}^+$ and suppose that $\tau < \tau' < 2\tau - 1$ for some $\tau' \in \mathbb{Z}^+$.  Then there does not exist a uniform mapping from $\{1, 2, ... \tau\}$ to $\{1, 2, ..., \tau'\}$.
\end{corollary}
}%

\ifabstract
\later{

\begin{proof}
Suppose for the sake of contradiction that there does exist a uniform mapping from $\{1, 2, ... \tau\}$ to $\{1, 2, ..., \tau'\}$, say $M$.  Since $\tau < \tau'$ and $M$ is uniform, it must be the case that $M(1) \neq 1$ since otherwise $\tau * 1 = \tau = \tau * M(1) < \tau'$.  But, observe that $\sum\limits_{i=1}^{\tau-1} 1 = (\tau-1) < \tau$ but $\sum\limits_{i=1}^{\tau-1} M(1) \geq \sum\limits_{i=1}^{\tau-1}2=2(\tau-1)=2\tau-2\geq\tau'$.  This contradicts the assumption that $M$ is uniform.
\qed
\end{proof}

}%
\fi

\both{
\begin{corollary}\label{cor:finite-no-mapping}
For any $\tau \in \mathbb{Z}^+$, there are a finite number of $\tau' \in \mathbb{Z}^+$ with $\tau' > \tau$ such that a uniform mapping cannot be found from $\tau$ to $\tau'$.
\end{corollary}
}%

\ifabstract
\later{

\begin{proof}
By Corollary~\ref{cor:mapping-constant}, for any given $\tau$ and $\tau'$ a uniform mapping exists if and only if there exists a $c \in \mathbb{N}$ that satisfies $\frac{\tau' - 1}{\tau} < c < \frac{\tau'}{\tau - 1}$. Such a $c$ exists whenever $\frac{\tau'}{\tau - 1} > \frac{\tau' - 1}{\tau}$ and $\frac{\tau'}{\tau - 1} - \frac{\tau' - 1}{\tau} > 1$, which is true when $\tau + \tau' > 1$ and $\tau' > \tau^2 - 2\tau + 1$. Therefore, a uniform mapping can be found from any $\tau \ge 1$ to any $\tau' > \tau^2 - 2\tau + 1$.
qed
\end{proof}

}%
\fi

\both{
\begin{theorem}\label{thm:uniform-mapping}
Given $\tau,\tau' \in \mathbb{Z}^+$ with $\tau \le \tau'$, there exists an algorithm which runs in time $O(\log^2 \tau')$ and (1) determines whether or not a uniform mapping from $\tau$ to $\tau'$ exists, and (2) if so, produces that mapping.
\end{theorem}
}%

\ifabstract
\later{

\begin{proof}
By Corollary~\ref{cor:mapping-constant}, we know that we must simply determine whether or not there exists some constant $c \in \mathbb{N}$ such that $c(\tau-1) < \tau' \le c\tau$ to determine whether or not a uniform mapping exists.  If we find such a $c$, then a uniform mapping exists and we can define an almost linear uniform mapping using that $c$.  Therefore, we must find $c$ such that $c \ge \lfloor \tau'/\tau \rfloor$ and $c < \lfloor \tau'/(\tau-1) \rfloor$.  This is done by letting $c$ equal the floor of $\tau'$ divided by $\tau$ and determining if that $c$ is less than $\lfloor \tau'/(\tau-1) \rfloor$.  If not, no uniform mapping exists from $\tau$ to $\tau'$.  If so, one does and it is simply:

\begin{displaymath}
   M(x) = \left\{
     \begin{array}{lr}
       cx & : x \neq \tau\\
       \tau'  & : x = \tau.
     \end{array}
   \right.
\end{displaymath}

Finally, since $\tau \le \tau'$, the algorithm to determine $c$, and thus the mapping, requires only two division operations with $\tau$ and $\tau'$ and a comparison of the results, this can be done in time $O(\log^2 \tau')$.
\qed
\end{proof}

}%
\fi

The following corollary will be used later in the proof of Lemma~\ref{lem:StrongImpossible}.

\both{
\begin{corollary}\label{cor:no-uniform-mapping-implic}
Given $\tau,\tau' \in \mathbb{N}$ such that $1 < \tau < \tau'$, if no uniform mapping exists from $\tau$ to $\tau'$, then $(\tau-1)\lceil \frac{\tau'}{\tau} \rceil \ge \tau'$.
\end{corollary}
}%

\ifabstract
\later{

\begin{proof}
To prove Corollary~\ref{cor:no-uniform-mapping-implic}, we assume the opposite and prove by contradiction.  Therefore, assume that, given $\tau,\tau' \in \mathbb{N}$ such that $1 < \tau < \tau'$ and no uniform mapping exists from $\tau$ to $\tau'$, that $(\tau-1)\lceil \frac{\tau'}{\tau} \rceil < \tau'$.  We first note that $\tau' = \tau(\frac{\tau'}{\tau}) \le \tau \lceil \frac{\tau'}{\tau} \rceil$.  We thus have the inequality  $(\tau-1)\lceil \frac{\tau'}{\tau} \rceil < \tau' \le \tau \lceil \frac{\tau'}{\tau} \rceil$.  By Corollary~\ref{cor:mapping-constant}, we know that if there exists a constant $c \in \mathbb{N}$ such that $c(\tau-1) < \tau' \le c\tau$, then there exists a uniform mapping from $\tau$ to $\tau'$.  By setting $c = \lceil \frac{\tau'}{\tau} \rceil$, we see that there must be a uniform mapping between $\tau$ and $\tau'$, which is a contradiction, and thus Corollary~\ref{cor:no-uniform-mapping-implic} is proven.
\qed
\end{proof}

}%
\fi

%% file: StrongPossible.tex
\section{Strong Simulation Via Uniform Mappings}\label{sec:strong-sim-with-mappings}

\ifabstract
\later{
\section{Proofs from Section~\ref{sec:strong-sim-with-mappings}: Strong Simulation Via Uniform Mappings}

}%
\fi

In this section, we provide positive results showing that for any pair of temperatures $\tau,\tau' \in \mathbb{Z}^+$ such that $\tau < \tau'$ and there is a uniform mapping from $\tau$ to $\tau'$, then there exists a tile set $U_{\tau'}$ which is intrinsically universal at temperature $\tau'$ for the class of all 2HAM systems at temperature $\tau$.

\both{
\begin{lemma}\label{lem:strong-sim}
Let $\tau,\tau' \in \mathbb{Z}^+$ with $\tau < \tau'$, such that there exists a uniform mapping $M$ from $\tau$ to $\tau'$, and let $\calT = (T,S,\tau)$, be an arbitrary 2HAM system at temperature $\tau$.  Then, there exists $\calT' = (T',S',\tau')$ such that $\calT'$ strongly simulates $\calT$.
\end{lemma}
}%

To prove Lemma~\ref{lem:strong-sim}, we show how to create $T'$ from $T$ by using the mapping $M$.  $T'$ is essentially identical to $T$, but for each glue $g$ on a tile in $T$, if its strength is given by the function $\texttt{str}(g)$, then the strength of that glue in $T'$ is equal to $M(\texttt{str}(g))$.  Due to the properties of a uniform mapping, we show that if and only if a multiset of glues on a pair of supertiles over $T$ allow those supertiles to bind in $\calT$, the mapped glues over supertiles in $T'$ will allow the equivalent supertiles in $\calT'$ to bind.  Thus, $\calT'$ will correctly strongly simulate $\calT$.

\ifabstract
\later{
\begin{proof}
To prove Lemma~\ref{lem:strong-sim}, we create the tile set $T'$ as follows.  For each $t \in T$, create a tile $t'$ which is identical in all properties except for the strengths of the glues.  For each glue strength, if the strength of that glue on $t$ was $s$, then make its strength on $t'$ equal to $M(s)$.  To create $S'$, we simply create and add a copy, $\alpha'$, of each $\alpha \in S$ by swapping the tiles of $T$ for the corresponding tiles in $T'$. $\alpha'$ is guaranteed to be $\tau'$-stable if and only if $\alpha$ is $\tau$-stable because of the following. Every cut across a supertile $\alpha'$ will break some multiset $G'$ of glues in $T'$. Similarly, the same cut across $\alpha$ will break some multiset $G$ of glues in $T$.  Let $n_1 = \Sigma_{g\in G} \texttt{str}(g)$ (where $\texttt{str}(x)$ is the function which returns the strength of glue $x$).  By the definition of our tile set $T'$ and the assembly $\alpha'$ based on $\alpha$, we know that $\Sigma_{g' \in G'} \texttt{str}(g') = \Sigma_{g \in G} M(\texttt{str}(g))$, and we set $n_2$ equal to this summation. By the definition of $M$ we know that $n_1 \ge \tau$ if and only if $n_2 \ge \tau'$, and thus $\alpha'$ is $\tau'$-stable if and only if $\alpha$ is $\tau$-stable.

To prove that $\calT' = (T',S',\tau')$ strongly simulates $\calT$, we simply let the scale factor of the simulation be $1$ and the $1$-block representation function $R: B^{T'}_1 \rightarrow T$ map each tile of $T'$ directly to the single, unique tile which it is a (glue-strength-modified) copy of. (Note that $R$ is a bijection here.)  Now, we start with the base case of singleton tiles of $T'$ and the assemblies in $S'$ (i.e. $T' \cup S'$), and compare their behavior to the singleton tiles of $T$ and assemblies of $S$ (i.e. $T \cup S$) to which they map.  Let $\ta'$ be any element of $(T' \cup S')$, and $\ta$ be the corresponding element of $T \cup S$ (i.e. $\tilde{R}(\ta')$).  By the definition of the tiles in $T'$, all glues exposed on the perimeter of $\ta'$ are identical in type and location to those on $\ta$, and the strength $s'$ of each is equal to $M(s)$ of the strength $s$ of the corresponding glue on $\ta$.  Let $\tb \in (T \cup S)$ be such that $\ta$ can $\tau$-stably combine with $\tb$ to produce $\tg \in \prodasm{\calT}$, and let $G$ be the (multi)set of glues which bind between $\ta$ and $\tb$.  Clearly $\Sigma_{g \in G} \texttt{str}(g) \ge \tau$, meaning the summation of the strengths of the binding glues is $\ge \tau$ since $\tg$ is $\tau$-stable.  By letting $\tb'$ be the supertile such that $\tilde{R}(\tb') = \tb$, we can verify that $\ta'$ can $\tau'$-stably combine with $\tb'$ to form $\tg'$, where $\tilde{R}(\tg') = \tg$, since by the uniform mapping $M$ from $\tau$ to $\tau'$ and our assignment of glue strength values in $T'$ we know that $\ta'$ and $\tb'$ would have the same (multi)set $G$ of binding glues (modulo the modified strengths) and that $\Sigma_{g \in G} M(\texttt{str}(g)) \ge \tau'$ by definition of the uniform mapping $M$.  Furthermore, for any (multi)set of glues $G_2$ over the glues in $T$ such that $\Sigma_{g\in G_2} \texttt{str}(g) < \tau$, for the corresponding set $G_2'$ over the glues in $T'$, $\Sigma_{g\in G_2'} M(\texttt{str}(g)) < \tau'$, again by the definition of the uniform mapping $M$.  Thus it is shown that a pair of supertiles $\ta',\tb' \in (T' \cup S')$ will be able to $\tau'$-stably combine if and only if $\ta,\tb \in (T \cup S)$, where $\tilde{R}(\ta') = \ta$ and $\tilde{R}(\tb') = \tb$, can $\tau$-stably combine, completing the base case.  This argument can then be applied recursively to all producible supertiles in both systems, so that for every set of producible supertiles $\ta,\tb,\tg \in \prodasm{\calT}$, $\ta$ can $\tau$-stably combine with $\tb$ to form $\tg$ if and only if for $\ta',\tb',\tg' \in \prodasm{\calT'}$ where $\tilde{R}(\ta') = \ta$, $\tilde{R}(\tb') = \tb$, and $\tilde{R}(\tg') = \tg$, $\ta'$ can $\tau'$-stably combine with $\tb'$ to form $\tg'$.  Therefore, $\calT$ and $\calT'$ have equivalent productions, $\calT$ follows $\calT'$, and $\calT'$ strongly models $\calT$, and thus $\calT'$ strongly simulates $\calT$.
\qed
\end{proof}
}%
\fi

Lemma~\ref{lem:strong-sim} shows that as long as there is a uniform mapping between two temperatures, for each system at the lower temperature there exists a system at the higher temperature which can strongly simulate it.  Furthermore, Corollaries~\ref{cor:no-uniform-mapping} and \ref{cor:finite-no-mapping} show us that there are only a very few temperatures greater than a given $\tau$ for which a uniform mapping does not exist.  Theorem~\ref{thm:uniform-mapping} tells us that we can efficiently find a uniform mapping $M$ if one exists, and by the proof of Lemma~\ref{lem:strong-sim} we can also see that the generation of the simulating system merely requires $M$ and time linear in the size of the system to be simulated.  We now show that such a strongly simulating system can be created for a tile set which is intrinsically universal for systems at $\tau$, resulting in a tile set which is IU for systems at temperature $\tau$ while strongly simulating them at $\tau'$.

\both{
\begin{theorem}\label{thm:IU-down}
Let $\tau,\tau' \in \mathbb{Z}^+$ with $1 < \tau < \tau'$, such that there exists a uniform mapping $M$ from $\tau$ to $\tau'$.  Then there exists a tile set $U_{\tau'}$ which is intrinsically universal for the class of all 2HAM systems at temperature $\tau$, such that the simulating systems using $U_{\tau'}$ are at temperature $\tau'$.
\end{theorem}
}%

\ifabstract
\later{
\begin{proof}
By Theorem 3 of \cite{2HAMIU}, for each $\tau \ge 2$ there exists a tile set $U_{\tau}$ such that $U_{\tau}$ is intrinsically universal for the class of all 2HAM systems at temperature $\tau$, such that the simulating systems using $U_{\tau}$ are at temperature $\tau$.  Since there exists a uniform mapping $M$ from $\tau$ to $\tau'$, we simply apply the technique used in the proof of Lemma~\ref{lem:strong-sim} to generate the tile set $U_{\tau'}$ from $U_{\tau}$ and the representation function $R': B^{U_{\tau'}}_1 \rightarrow U_{\tau}$ (note that $R'$ is again a bijection here).  The $U_{\tau'}$ thus generated is intrinsically universal for $\tau$ as follows.  Let $\calT = (T,S,\tau)$ be an arbitrary 2HAM system at temperature $\tau$. Let $\mathcal{U}_{\calT} = (U_{\tau},S_{\calT},\tau)$ be the temperature $\tau$ system which uses $U_{\tau}$ to simulate $\calT$ at scale factor $m$, and let $R: B^{U_{\tau}}_m \dashrightarrow T$ be the representation function mapping blocks of tiles from $U_{\tau}$ to tiles of $T$. We now define the system $\mathcal{U}'_{\calT} = (U_{\tau'}, S'_{\calT}, \tau')$ which uses the previously defined tile set $U_{\tau'}$, and since $U_{\tau'}$ is used to simulate $U_{\tau}$ at scale factor 1, we make the assemblies in $S'_{\calT}$ as exact copies of the assemblies of $S_{\calT}$ but with each tile $t \in U_{\tau}$ replaced by the tile returned by $R'^{-1}(t)$.  We make the scale factor for the simulation of $\calT$ by $\mathcal{U'}_{\calT}$ to be $m$ and the representation function $R'' = R \circ R'$.  Since $\mathcal{U'}_{\calT}$ strongly simulates (and therefore by definition also simulates) $\mathcal{U}_{\calT}$ at scale factor 1 under $R'$, and $\mathcal{U}_{\calT}$ simulates $\calT$ at scale factor $m$ under $R$, then $\mathcal{U'}_{\calT}$ simulates $\calT$ at scale factor $m$ under $R''$.
\qed
\end{proof}
}%
\fi

The proof of Theorem~\ref{thm:IU-down} simply makes use of the result of \cite{2HAMIU} showing that for the class of systems at each temperature $\tau \ge 2$, there exists a tile set which is IU for that class.  That IU tile set simulates at temperature $\tau$, so we use Lemma~\ref{lem:strong-sim} to show that for $\tau' > \tau$ where a uniform mapping exists from $\tau$ to $\tau'$, we can make a strongly simulating tile set at temperature $\tau'$ for the tile set which is IU for $\tau$ systems.

Note that the results of \cite{2HAMIU} provide for a variety of tile sets for each $\tau > 1$ such that each is IU for that $\tau$.  These tile sets provide for a variety of tradeoffs in scale factor, tile set size, and number of seed assemblies. Any such tile set $U_{\tau}$ can be used to create the tile set $U_{\tau'}$ from Theorem~\ref{thm:IU-down} to achieve the same tradeoffs since the simulation of $U_{\tau}$ by $U_{\tau'}$ is at scale factor 1 and there is a bijective mapping of tile types from $U_{\tau'}$ to whichever $U_{\tau}$ is chosen. Furthermore, an IU tile set at temperature $\tau$ can be chosen which is IU in terms of either strong simulation or standard simulation, and by those definitions the result still holds. 

%% file: StrongImpossible.tex
\section{Impossibility of Strong Simulation at Higher Temperatures}\label{sec:impossibility}

Intuitively, it may appear that the class of systems at higher temperatures is more ``powerful'' than the class of systems at lower temperatures.  In this section, we show that this is not strictly the case. Here we present a sketch of the proof by giving an example of a tile set $U$ such that there exists a 2HAM TAS $\mathcal{T} = (T, S, 3)$ such that for any initial configuration $S_{\mathcal{T}}$ over $U$, the 2HAM TAS $\mathcal{U} = (U, S_{\mathcal{T}}, 4)$ does not strongly simulate $\mathcal{T}$. This gives an intuitive idea of the general proof which can be found in Section~\ref{sec:impossibility-proof}.

\ifabstract
\later{
\section{Proofs from Section~\ref{sec:impossibility}: Impossibility of Strong Simulation at Higher Temperatures}\label{sec:impossibility-proof}
In this section we present a more formal version of the proof of Lemma~\ref{lem:StrongImpossible}.\\
} %
\fi

\begin{theorem} \label{lem:StrongImpossible}
Let $\tau, \tau' \in \mathbb{N}$ be such that (1) $2 < \tau < \tau'$ and (2) there does not exist a uniform mapping from $\tau$ to $\tau'$.  For every tile set $U$, there exists a 2HAM TAS $\mathcal{T} = (T, S, \tau)$ such that for any initial configuration $S_{\mathcal{T}}$ over $U$, the 2HAM TAS $\mathcal{U} = (U, S_{\mathcal{T}}, \tau')$ does not strongly simulate $\mathcal{T}$.
\end{theorem}

\both{
\noindent {\em Proof:}
As in~\cite{2HAMIU}, the idea behind this proof is to use Definitions~\ref{scott-defn:alt-equiv-dynamic-t-to-s} and~\ref{scott-defn:alt-equiv-dyanmic-s-to-t-strong} in order to show two producible supertiles in $\mathcal{T}$ which cannot bind due to insufficient strength, but whose simulating supertiles in $\mathcal{U}$ can combine.   This will contradict the definition of simulation.  A large part of the terminology and notation in this proof are borrowed from~\cite{2HAMIU}.
\begin{figure}[htp]

\begin{center}
\includegraphics[width=4.5in]{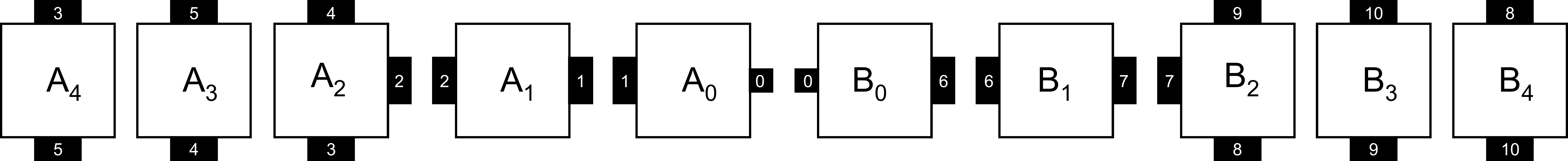}
\caption{The tile set for the proof of Theorem~\ref{lem:StrongImpossible}.  Black rectangles represent strength-$\tau$ glues (labeled 1-8), and black squares represent the strength-1 glue (labeled 0).}
\label{fig:ladders-tile-set}
\end{center}

\end{figure}

Our proof is by contradiction. Therefore, suppose, for the sake of obtaining a contradiction, that there exists an intrinsically universal tile set $U$ such that, for any 2HAM TAS $\mathcal{T} = (T,S,\tau)$, there exists an initial configuration $S_{\mathcal{T}}$ and $\tau'\geq \tau$, such that $\mathcal{U} = \left(U, S_{\mathcal{T}}, \tau'\right)$ strongly simulates $\mathcal{T}$ and there does not exist a uniform mapping from $\tau$ to $\tau'$. Define $\mathcal{T} = (T, \tau)$ where $T$ is the tile set defined in Figure~\ref{fig:ladders-tile-set}, the default initial state is used, and $\tau > 2$.  Let $\mathcal{U} = \left(U, S_{\mathcal{T}}, \tau'\right)$ be the temperature $\tau'\geq \tau$ 2HAM system, which uses tile set $U$ and initial configuration $S_{\mathcal{T}}$ (depending on $\mathcal{T}$) to strongly simulate $\calT$ at scale factor $m$. Let $\tilde{R}$ denote the supertile representation function that testifies to the fact that $\mathcal{U}$ strongly simulates $\mathcal{T}$.
}%

\begin{figure}[htp]
\begin{center}
\includegraphics[width=2in]{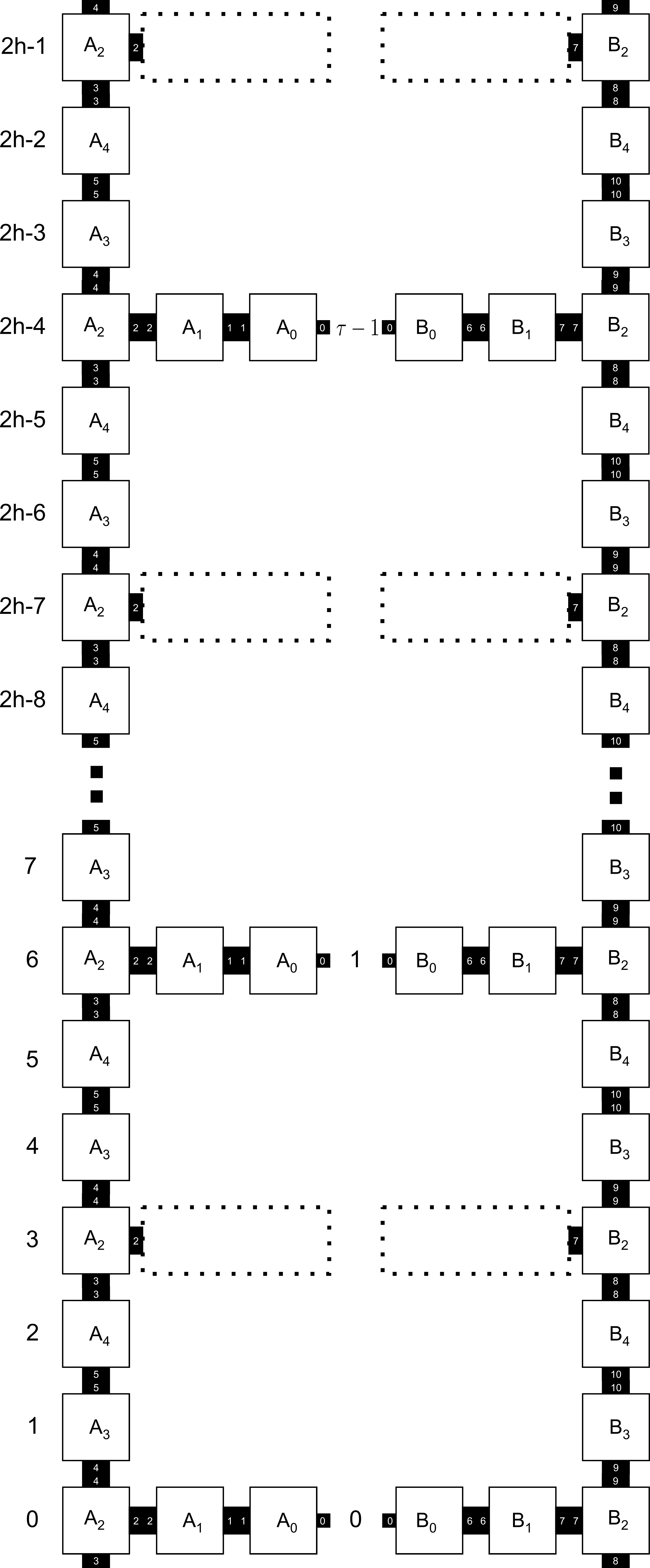}
\caption{Example half-ladders with $\tau$ rungs.}
\label{fig:ladders}
\end{center}
\end{figure}

\both{
We say that a supertile $\tilde{l} \in \prodasm{\mathcal{T}}$ is a $d$-rung \emph{left half-ladder} of height $h \in \mathbb{N}$ if it contains $h$ tiles of the type $A_2$ and $h-1$ tiles of type $A_3$, arranged in a vertical column, plus $d$ tiles each of the types $A_1$ and $A_0$ for $d \in \mathbb{N}$. 
(An example of a $\tau$-rung left half-ladder is shown on the left in Figure~\ref{fig:ladders}. The dotted lines show positions at which tiles of type $A_1$ and $A_0$ could potentially attach, but since a $\tau$-rung half-ladder has exactly $\tau$ of each, only $\tau$ such locations have tiles.)
Essentially, a $d$-rung left half-ladder consists of a single-tile-wide vertical column of height $2h-1$ with an $A_2$ tile at the bottom and top, and those in between alternating between $A_2$, $A_3$, and $A_4$ tiles. To the east of exactly $d$ of the $A_2$ tiles an $A_1$ tile is attached and to the east of each $A_1$ tile an $A_0$ tile is attached.
These $A_1$-$A_0$ pairs, collectively, form the $\tau$ \emph{rungs} of the left half-ladder. We enumerate the $A_2$ tiles appearing in $\tilde{l}$ from north to south and denote the $i^{th}$ $A_2$ tile by $A_{2,i}$. Thus, $A_{2,0}$ denotes the northernmost $A_2$ tile in $\tilde{l}$ and $A_{2,(d-1)}$ denotes the southernmost tile in $\tilde{l}$.  We can define $d$-rung \emph{right} half-ladders similarly. A $d$-rung \emph{right half-ladder} of height $h$ is defined exactly the same way but using the tile types $B_3$, $B_2$, $B_1$, and $B_0$ and with rungs growing to the left of the vertical column. The east glue of $A_0$ is a strength-$1$ glue matching the west glue of $B_0$.
}%

\begin{figure}[htp]
\vspace{-16pt}
\centering
  \subfloat[][]{%
        \label{fig:impossibility-example1}%
        \includegraphics[height=7in]{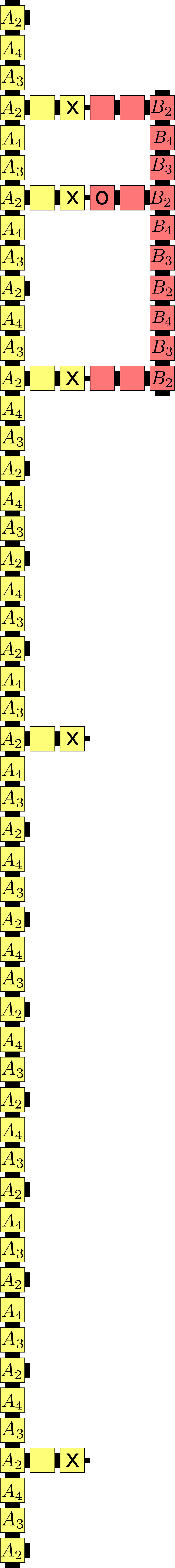}}
        \quad\quad\quad\quad\quad\quad\quad\quad
  \subfloat[][]{%
        \label{fig:impossibility-example2}%
        \includegraphics[height=7in]{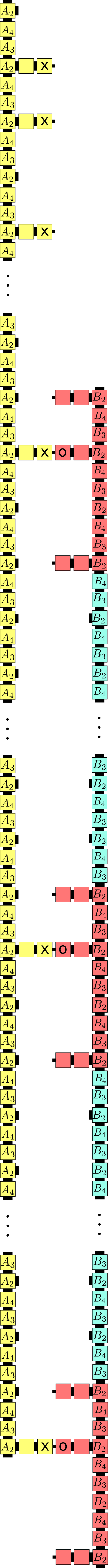}}
  \caption{The squares in this figure depict macrotiles which assemble in $\mathcal{U}$ and simulate tiles $\mathcal{T}$ when $\tau'=4$ and $\tau = 3$.}
  \label{fig:impossibility-examples}
\end{figure}

\both{
We say that a supertile consisting only of tiles of type $A_2$, $A_3$, and $A_4$ is a \emph{left bar} provided that the northernmost tile in the supertile is $A_4$ and the southernmost tile in the supertile is $A_3$.  The height of a bar is the number of $A_2$ tiles appearing in the bar.  We define a \emph{right bar} similarly.
} %
In the case where $\tau = 3$ and $\tau' = 4$, note that there does not exist a uniform mapping from $\tau$ to $\tau'$. Also, in this case, Figure~\ref{fig:impossibility-examples} shows the main idea of the proof of Theorem~\ref{lem:StrongImpossible}.

Consider the left half-ladder shown in Figure~\ref{fig:impossibility-example1}. We show that for sufficiently many rungs, some macrotile (labeled $x$) must repeat an arbitrary number of times. Therefore, for strong simulation, there must be a left half-ladder, $\tilde{l}'$, with rungs that contain these macrotiles. $\tilde{l}'$ is depicted by yellow tiles. By assumption, $\mathcal{T}$ is strongly simulated by $\mathcal{U}$, therefore, there must be a $3$ rung right half-ladder which we call $r_p'$ that binds to exactly three of the rungs of $\tilde{l}'$. $\tilde{r_p}'$ is depicted by red tiles. Note that because $\tau' > \tau$, it must be the case that some rung binds with strength at least $\lceil \frac{\tau'}{\tau}\rceil$ (we say that such a rung ``over-binds''.)  Moreover, we show that we can choose $x$ such that $x$ belongs to an ``over-binding'' rung and such that the distance between each consecutive macrotile $x$ is increasing.  Then, as depicted in Figure~\ref{fig:impossibility-example2}, we use the assumption of strong simulation to construct a right half-ladder which we call $\tilde{r}_{bar}'$ that consists of $\tau - 1$ copies of the supertile $\tilde{r}_p'$ bound to spacer macrotiles such that each copy of $\tilde{r}_p'$ is precisely and appropriately spaced. The tiles which bind between copies of $\tilde{r}_p'$ supertiles are depicted by blue tiles. Note that each $\tilde{r}_p'$ contains an ``over-binding'' rung. Then, the spacings of the $\tilde{r}_p'$ supertiles of $\tilde{r}_{bar}'$ are chosen so that only ``over-binding'' rungs attach to $\tilde{l}'$ and each ``over-binding'' rung attaches to a rung of $\tilde{l}'$ with at least strength $\lceil \frac{\tau'}{\tau}\rceil$. Finally, given the assumption that there is not a uniform mapping from $\tau$ to $\tau'$, it follows that $(\tau-1)\lceil \frac{\tau'}{\tau} \rceil \geq \tau'$. We then show that this implies that $\tilde{l}'$ and $\tilde{r}_{bar}'$ can bind in $\mathcal{U}$, but that $\tilde{R}(\tilde{l}')$ cannot stably bind to $\tilde{R}(\tilde{r}_{bar}')$. Thus, we arrive at a contradiction. It should be noted that the proof is not merely combinatorial and relies on arguing about the dynamics of $\mathcal{U}$, though we have not indicated that here. Please see Section~\ref{sec:impossibility-proof} for more detail.

\ifabstract
\later{

Let $LEFT \subseteq \prodasm{T}$ and $RIGHT \subseteq \prodasm{T}$ be the set of all left and right half-ladders of height $h$, respectively. Note that there are $\binom{h}{\tau}$ $\tau$-rung half-ladders of height $h$ in $LEFT$ ($RIGHT$). Define, for each $\tilde{l} \in LEFT$, the \emph{mirror image} of $\tilde{l}$ as the supertile $\bar{\tilde{l}} \in RIGHT$ such that $\bar{\tilde{l}}$ has rungs at the same positions as $\tilde{l}$.

For some $\tilde{l} \in LEFT$, we say that $\tilde{\hat{l}} \in \prodasm{U}$ is a \emph{simulator} left half-ladder of height $h$ if $\tilde{R}\left( \tilde{\hat{l}}\right) = \tilde{l}$. Note that $\tilde{\hat{l}}$ need not be unique, e.g., $\tilde{\hat{l}}$ and $\tilde{\hat{l}}'$ could differ by a single tile (the latter could have no simulation fuzz and the former could have one tile of fuzz) yet satisfy $\tilde{R}\left(\tilde{\hat{l}}'\right) = \tilde{l}$. The notation $C^{\tau}_{\tilde{\alpha},\tilde{\beta}}$ is defined as the set of all supertiles that result in the $\tau$-stable combination of the supertiles $\tilde{\alpha}$ and $\tilde{\beta}$.

For some $\tilde{\hat{r}} \in \prodasm{U}$, we say that $\tilde{\hat{r}}$ is a \emph{mate} of $\tilde{\hat{l}}$ if $\tilde{R}\left(\tilde{\hat{r}}\right) = \tilde{r} \in RIGHT$, where $\tilde{r} = \bar{\tilde{l}}$, $C^{\tau}_{\tilde{l},\tilde{r}} \ne \emptyset$ (they combine in $\mathcal{T}$), and $C^{\tau'}_{\hat{\tilde{l}},\hat{\tilde{r}}} \ne \emptyset$ (they combine in~$\mathcal{U}$). For a simulator left half-ladder $\tilde{\hat{l}}$, we say that $\tilde{\hat{l}}$ is \emph{combinable} if $\tilde{\hat{l}}$ has a mate.

Let $\tilde{\alpha'} \in \prodasm{U}$ be such that $\tilde{R}(\tilde{\alpha'})=\tilde{\alpha}$ for some $\tilde{\alpha} \in \prodasm{T}$.  We say that a simulator supertile $\tilde{\alpha'}$ is a \emph{maximal simulator supertile} provided that the attachment of any other supertile in $\prodasm{U}$ implies $\tilde{R}(\tilde{\alpha'}) \neq \tilde{\alpha}$.

\begin{observation} \label{ob:max}
Let $\ta, \tb \in \prodasm{T}$, and suppose that $\tilde{\gamma} \in C^\tau_{\ta,\tb}$.  Furthermore, let $\ta', \tb' \in \prodasm{U}$ be such that $\tilde{R}(\ta')=\ta$, $\tilde{R}(\tb')=\tb$, and suppose that $\ta'$ and $\tb'$ are both maximal simulator supertiles.  Then, there exists $\tilde{\gamma}'\in C^\tau_{\ta',\tb'}$ such that $\tilde{R}(\tilde{\gamma}')=\tilde{\gamma}$.
\end{observation}

This observation follows directly from Definition~\ref{scott-defn:alt-strong-simulate}, specifically its requirement that $\calU$ strongly model $\calT$ (Definition~\ref{scott-defn:alt-equiv-dyanmic-s-to-t-strong}).

Let $k = (|U|+1)^{4m^2}$, which is the number of ways to tile a neighborhood of four $m \times m$ squares (i.e. 4 $m$-block assemblies, or macrotiles, in the simulator) from a set of $|U|$ distinct tile types and possibly empty positions.  Let $\tilde{l} \in \prodasm{T}$ be the $(2\tau-1)(k+1)$-rung left half-ladder of height $2^{(2\tau-1)(k+1)}+2$ such that there is a rung protruding from every $A_{2, 2^i}$ tile for $i \in [0, 2(\tau-1)(k+1)-1]$.  Let $\tilde{l}' \in \prodasm{U}$ be such that $\tilde{R}(\tilde{l}')=\tilde{l}$ and $\tilde{l}'$ is a maximal simulator supertile.

\begin{observation}
There are at least $(2\tau-1)$ neighborhoods in $\tilde{l}'$ that map to tiles of $A_0$ (plus any additional simulator fuzz that connects to simulated $A_0$ tiles) which are tiled the same.
\end{observation}

To see why this observation holds, note that for every $(k+1)$ rung configurations, there must exist at least two configurations which are the same (since there are only $k$ ways to tile neighborhoods that map to $A_0$ plus any additional simulator fuzz that connects to simulated $A_0$ tiles).  Now, note that this implies for every $(2\tau-1)(k+1)$ rung configurations, there must exist at least $(2\tau-1)$ configurations which are the same.  Let $\mathcal{X}$ be the ordered list of such rung configurations, such that $x_0 \in \mathcal{X}$ is the northernmost.
If a rung contains the tile denoted $x_i$, we say that it is the \emph{$x_i$ rung}.  In addition, we say that any $x_i$ rung is an $x$ rung.

\begin{figure}[htp]
\begin{center}
\includegraphics[height=8.0in]{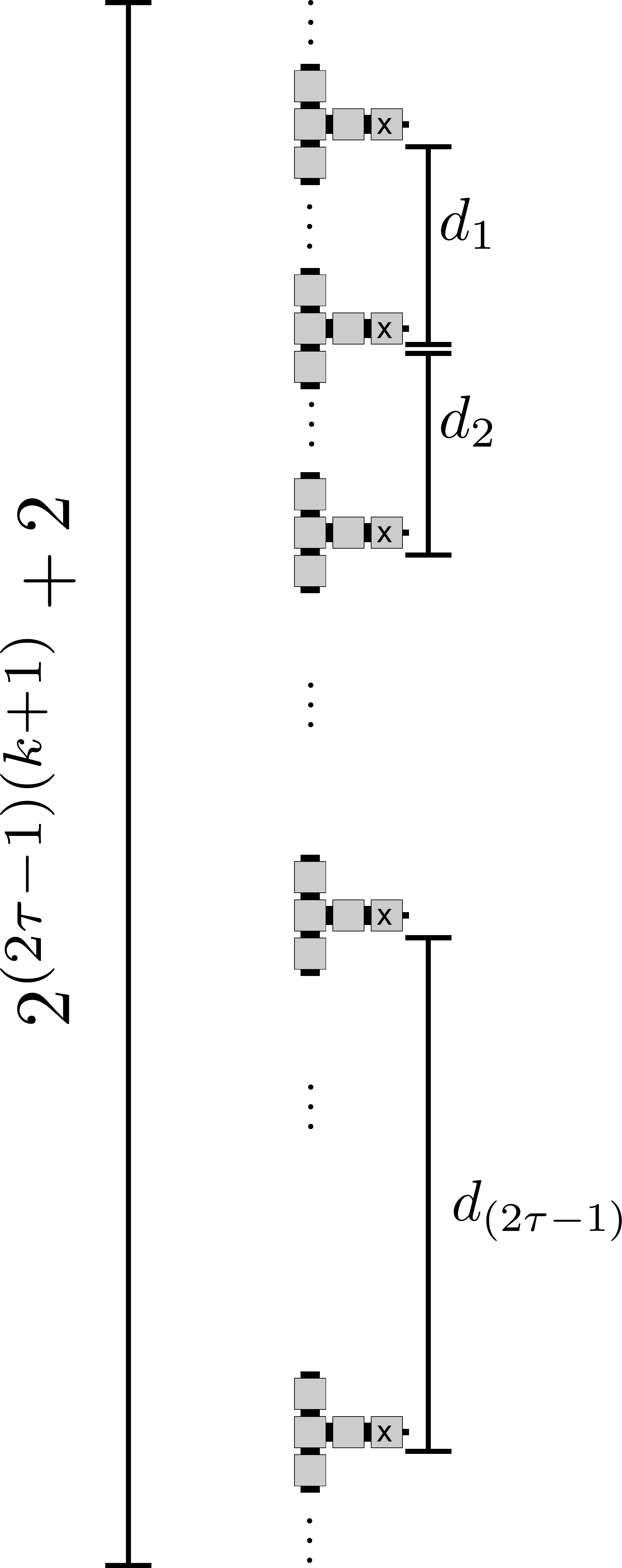}
\caption{The left half-ladder which we refer to as $\tilde{l}$.}
\label{fig:l}
\end{center}
\end{figure}

Let $\tilde{r_p} \in \prodasm{T}$ be the $\tau$-rung right half-ladder of minimal height such that its rungs are spaced the same as the first $\tau$ rungs of $\tilde{l}$.  That is $\tilde{r_p}$ has $\tau$ rungs which are capable of aligning with the northernmost $\tau$ $x$ rungs of $\tilde{l}$.  Define $\tilde{r_p}' \in \prodasm{U}$ to be a maximal simulator supertile such that $\tilde{R}(\tilde{r_p}')=\tilde{r_p}$.

Notice that $\tilde{r_p}$ and $\tilde{l}$ can combine as shown in Figure~\ref{fig:lr}.  Denote this supertile by $\tilde{f_p}$.  Let $\tilde{f_p}'$ be such that $\tilde{R}(\tilde{f_p}')=\tilde{f_p}$ and $\tilde{f_p}' \in C^\tau_{\tilde{l}',\tilde{r_p'}}$.  Now, observe that since there does not exist a uniform mapping between $\tau$ and $\tau'$, there exists at least one neighborhood in $\tilde{r_p'}$ that maps to a tile of type $B_0$ (plus any additional simulator fuzz that connects to simulated $B_0$ tiles) which contributes a binding strength of at least $\lceil \frac{\tau'}{\tau} \rceil$.  We denote the tiles in $\tilde{r_p} \in \prodasm{T}$ whose corresponding neighborhoods in $\tilde{r_p}' \in \prodasm{U}$ contribute strength at least $\lceil \frac{\tau'}{\tau} \rceil$ by $o$.  If a rung contains, an $o$ tile, we say that it is an \emph{$o$ rung}.

\begin{figure}[htp]
\begin{center}
\includegraphics[height=2.0in]{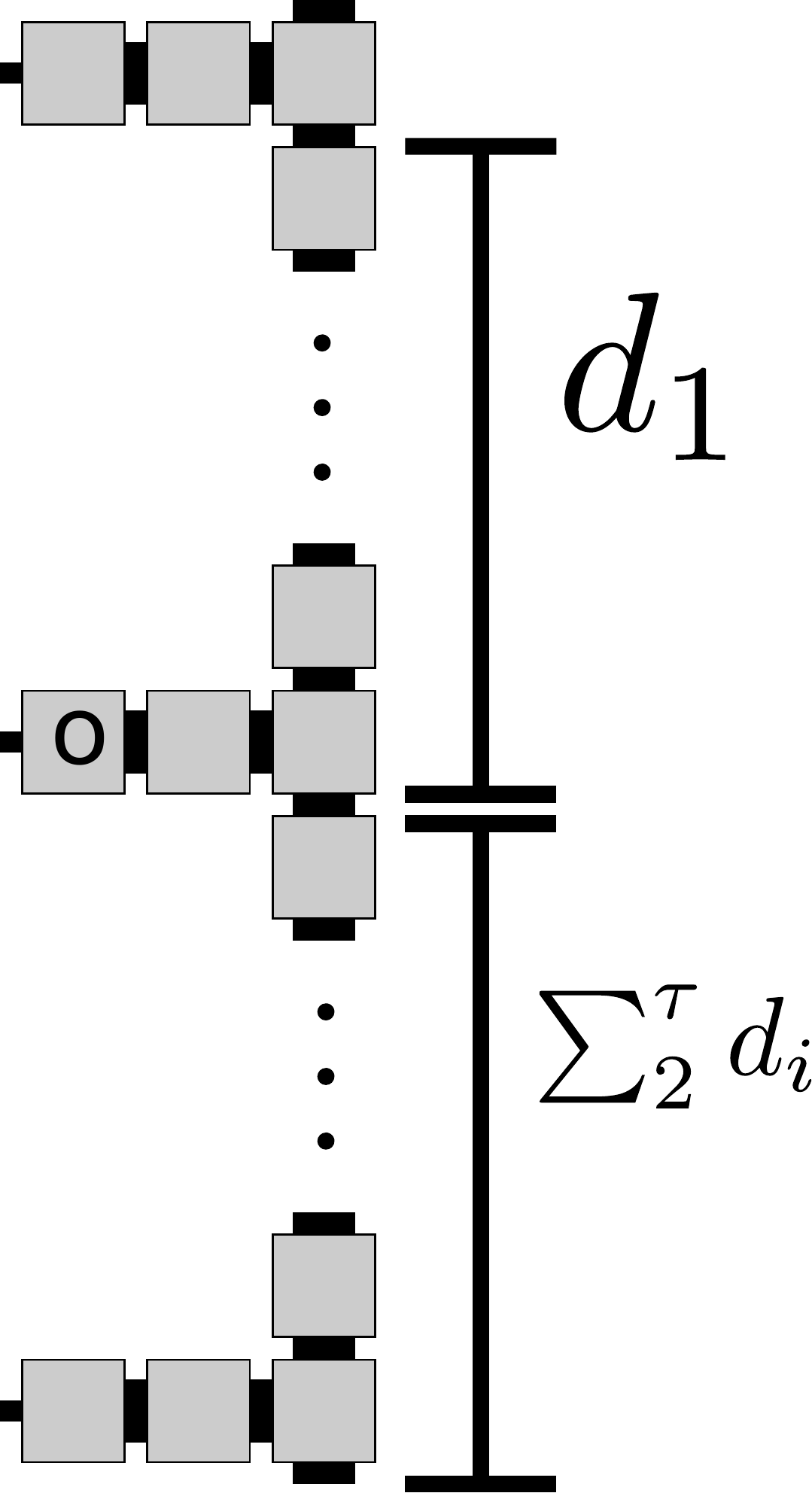}
\caption{The right half-ladder which we refer to as $\tilde{r}$.}
\label{fig:r}
\end{center}
\end{figure}

\begin{figure}[htp]
\begin{center}
\includegraphics[height=8.0in]{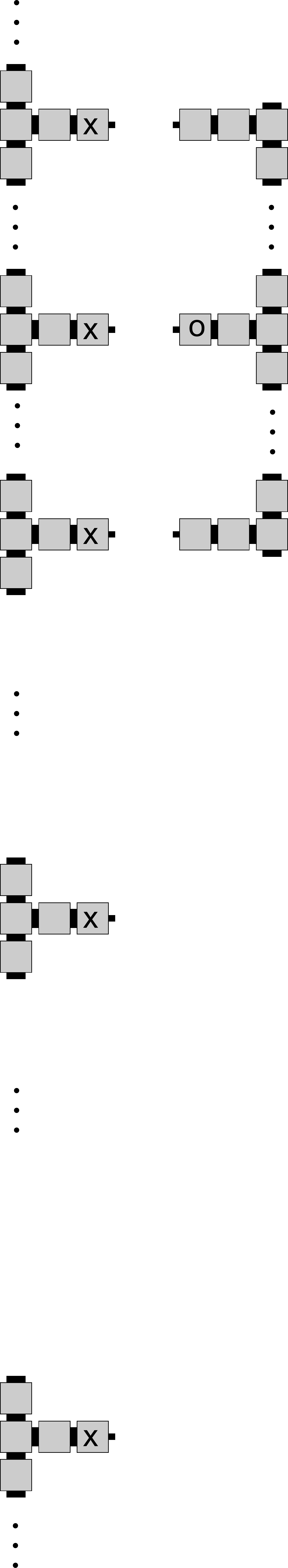}
\caption{The two supertiles $\tilde{l}$ and $\tilde{r}$ combining.}
\label{fig:lr}
\end{center}
\end{figure}

Let $h_r$ be the height of $r_p$, and let $d_i$ be the sum of all the $A_2$ tiles that lie between the rung which contains $x_i$ and the rung which contains $x_{i-1}$ including the boundary tiles.   We denote by $\tilde{b_i}$ the bars in $\prodasm{T}$ of length $d_i-h_r$ for $i \in [\tau+1, 2\tau-1]$.  Intuitively, the length of these bars is such that whenever $\tilde{r_p}$ supertiles attach to the north and south of the bar $b_i$ and an $o$ rung of the north $\tilde{r_p}$ supertile is aligned with the rung which contains $x_{i-1}$, the southern $\tilde{r_p}$ supertile will have an $o$ rung which aligns with the rung $x_i$.  Let $\tilde{b_i}' \in \prodasm{T}$ be a maximal simulator supertile such that $\tilde{R}(\tilde{b_i}') = \tilde{b_i}$.

Let $\tilde{r_i} \in \prodasm{T}$ be the supertile composed of $\tilde{b_i}$ attached to the southernmost tile of $\tilde{r_p}$(see Figure~\ref{fig:ri}).  Define $\tilde{r_i}' \in \prodasm{U}$ to be the set of supertiles such that 1) $\tilde{R}(\tilde{r_i}') = \tilde{r_i}$ and 2) $\tilde{r_i}' \in C^\tau_{\tilde{r_p}',\tilde{b_i}'}$.  We know that such an $\tilde{r_i}'$ exists because we defined $\tilde{r_p}$ and $\tilde{b_i}'$ to be maximal.

\begin{figure}[htp]
\begin{center}
\includegraphics[height=4.0in]{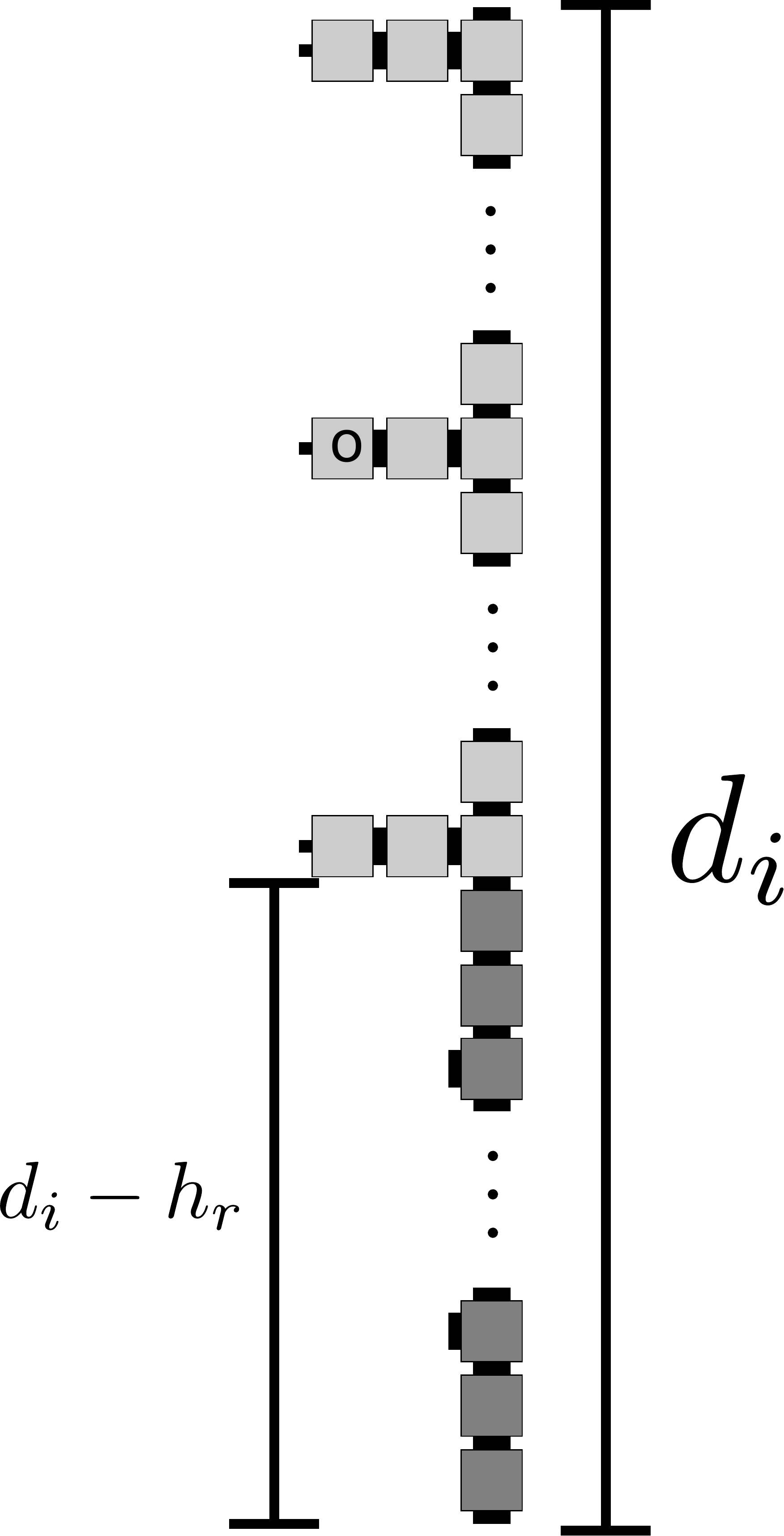}
\caption{The supertile $\tilde{r_{i}}$ formed by combining $\tilde{r_p}$ and a ``bar'' of length $\tilde{d_i}$ for $i \in [\tau, 2\tau-1]$.}
\label{fig:ri}
\end{center}
\end{figure}

Let $\tilde{r_{bar}} \in \prodasm{T}$ be the right half-ladder formed by combining the $\tau-1$ distinct $\tilde{r_i}$ supertiles so that $\tilde{r_{\tau}}$ is the northern most supertile in $\tilde{r_{bar}}$, $\tilde{r_{\tau+1}}$ is attached to the south of $\tilde{r_{\tau}}$, and in general $\tilde{r_{\tau+i+1}}$ is attached to the south of $\tilde{r_{\tau+i}}$ for $i \in [0, \tau-1]$.  Furthermore, $\tilde{r_{bar}}$ is such that $\tilde{r_p}$ attaches to the south of $\tilde{r_{2\tau+1}}$ and nothing else attaches below it.  See Figure~\ref{fig:rbar} for a depiction of $\tilde{r_{bar}}$.

Let $\tilde{r_{bar}}' \in \prodasm{U}$ be a supertile such that $\tilde{R}(\tilde{r_{bar}}') = \tilde{r_{bar}}$ and $\tilde{r_{bar}}'=\res{\tilde{r_i}' \mid \tau + 1 \leq i \leq 2\tau}$ where we define $\tilde{r_{2\tau}}$ to be the supertile $\tilde{r_p}$ and all other $\tilde{r_i}$ to be as above.  Since we defined $\tilde{r_p}'$ and $\tilde{b_i}$ to be maximal, it follows from a straight forward extension of Observation~\ref{ob:max} that this is a valid assembly sequence.

\begin{figure}[htp]
\begin{center}
\includegraphics[height=8.0in]{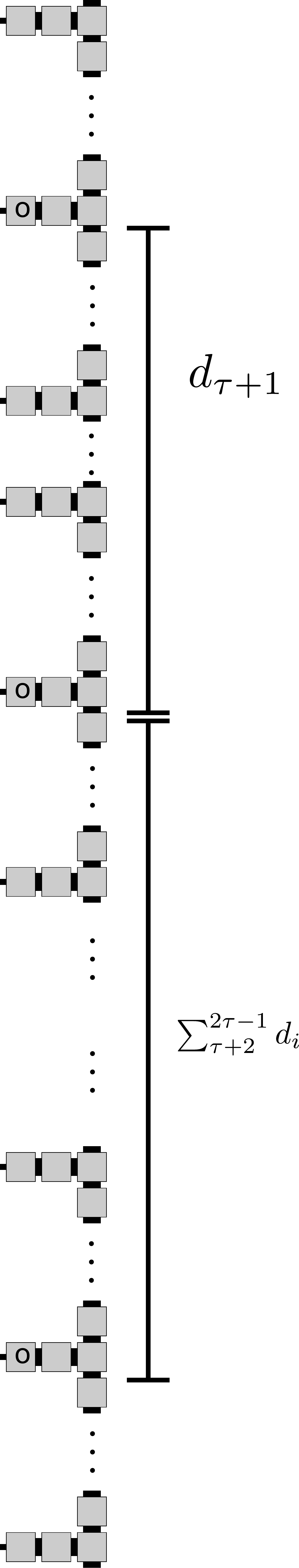}
\caption{The supertile $r_{bar}$ formed by combining the supertiles $\tilde{r_1}, \tilde{r_2}, ..., \tilde{r_{\tau-1}}$.}
\label{fig:rbar}
\end{center}
\end{figure}

Let $\tilde{\alpha_x}' \in \prodasm{U}$ be the supertile formed by attaching $\tilde{r_{bar}}'$ to the bottom $\tau-1$ $x$ rungs of $\tilde{l}'$ so that the $o$ rungs of $\tilde{r_{bar}}'$ align with the $x$ rungs of $\tilde{l}'$.  The supertile which $\tilde{\alpha_x}$ maps to is shown in Figure~\ref{fig:lrbar} (note that this supertile is not $\tau$-stable).

Let $\tilde{\alpha_x}$ be the supertile formed by attaching $\tilde{r_{bar}}$ to the bottom $\tau-1$ rungs of $\tilde{l}'$ so that the $o$ rungs of $\tilde{r_{bar}}$ align with the $x$ rungs of $\tilde{l}$ (see Figure~\ref{fig:lrbar}).  Note that $\tilde{\alpha_x} \notin \prodasm{T}$.  Let $\tilde{\alpha_x}'$ be a supertile made from tiles in $U$ which is such that $\tilde{R}(\tilde{\alpha_x}')=\tilde{\alpha_x}$.

\begin{claim}
Let $\tilde{l}'$, $\tilde{r_{bar}}'$, and $\tilde{\alpha_x}'$ be the supertiles described above.  Then $\tilde{\alpha_x}' \in C^\tau_{\tilde{l}',\tilde{r_{bar}}'}$.
\end{claim}

Before we begin the proof of the claim, we introduce some notation.  Let $r_p' \in \tilde{r_p}'$.  Denote some subassembly in $r_p'$ which maps to an $o$ rung in $r_p$ by $\alpha_{xr}$.  Now let $l' \in \tilde{l}'$ be the assembly which combines with $r_p'$ to form a member of $\tilde{f_p}'$.  Denote the subassembly in $l'$ that maps to the $x$ rung which binds with the previously mentioned $o$ rung by $\alpha_{xl}$.

To prove this claim, we show that there exists $l' \in \tilde{l}'$ and $r_{bar}' \in \tilde{r_{bar}}'$ such that 1) $l'$ and $r_{bar}'$ are disjoint and 2) $l' \cup r_{bar}' \equiv \alpha_x'$ is $\tau'$-stable and 3) $\alpha_x' \in \tilde{\alpha_x}'$.  Let $r_{bar}' \in \tilde{r_{bar}}'$, and choose $l' \in \tilde{l'}$ so that the northernmost $o$ rung of $r_{bar}'$ and the rung $x_{\tau}$ in $l'$ lie in the same position relative to each other as the $\alpha_{xr}$ and $\alpha_{xl}$ subassemblies mentioned above.  Note that since $\tilde{l}'$ and $\tilde{r_p}'$ combine, $\alpha_{xr}$ and $\alpha_{xl}$ are disjoint.

We say that a subassembly $\alpha$ is a \emph{rung of $r_{bar}'$} provided that $\alpha$ is a subassembly of $r_{bar}'$ and $\alpha$ maps onto a rung of $r_{bar}$.  We say $\alpha$ is an \emph{$o$ rung of $r_{bar}'$} provided that its is a rung of $r_{bar}'$ and it maps to an $o$ rung in $r_{bar}$.

We now argue that $r_{bar}'$ and $l'$ are disjoint assemblies.  Note that by the way Definition~\ref{scott-defn:alt-strong-simulate} restricts fuzz and the way we chose $r_{bar}'$ and $l'$ to have exponentially increasing distance between rungs, we need only to check that the neighborhoods of subassemblies in $l'$ mapping onto the rungs $x_{\tau+i}$ for $i \in [0, \tau-1]$ do not overlap the subassemblies in $r_{bar}'$ mapping onto rungs in $r_{bar}$.  Recall that by construction, the $o$ rungs of $r_{bar}'$ are all the same up to translation.  It follows from this fact and Lemma~\ref{lem:noShift} that the rungs of $r_{bar}'$ all lie in the same positions relative to the rung $x_{\tau+i}$ for $i \in [0, \tau-1]$. Now, recall that we chose $l'$ and $r_{bar}'$ so that the northernmost $o$ rung of $r_{bar}'$ and the rung $x_{\tau}$ in $l'$ lie in the same position relative to each other as the $\alpha_{xr}$ and $\alpha_{xl}$ subassemblies.  This along with the fact that $\alpha_{xr}$ and $\alpha_{xl}$ are disjoint, implies that $l'$ and $r_{bar}'$ are disjoint.

To see that $l' \cup r_{bar}'$ is $\tau'$-stable, we first recall that the sum of all the glues shared between $\alpha_{xl}$ and $\alpha_{xr}$ is at least $\lceil \frac{\tau'}{\tau} \rceil$.  As noted above, all of the rungs of $r_{bar}'$ all lie in the same positions relative to the rung $x_{\tau+i}$ for $i \in [0, \tau-1]$.  These two facts together imply that every rung on $r_{bar}'$ binds to the rungs on $l'$ with strength $\lceil \frac{\tau'}{\tau} \rceil$.  Consequently, $l'$ and $r_{bar}'$ will bind with total strength $(\tau-1)\lceil \frac{\tau'}{\tau} \rceil$ since there are $\tau-1$ rungs in $r_{bar}'$.  Note that by the assumption there is not a uniform mapping from $\tau$ to $\tau'$ and Corollary~\ref{cor:no-uniform-mapping-implic}, $(\tau-1)\lceil \frac{\tau'}{\tau} \rceil \geq \tau'$.

The third condition is straight forward to see.  Thus the proof of the claim is complete.

\begin{figure}[htp]
\begin{center}
\includegraphics[height=8.0in]{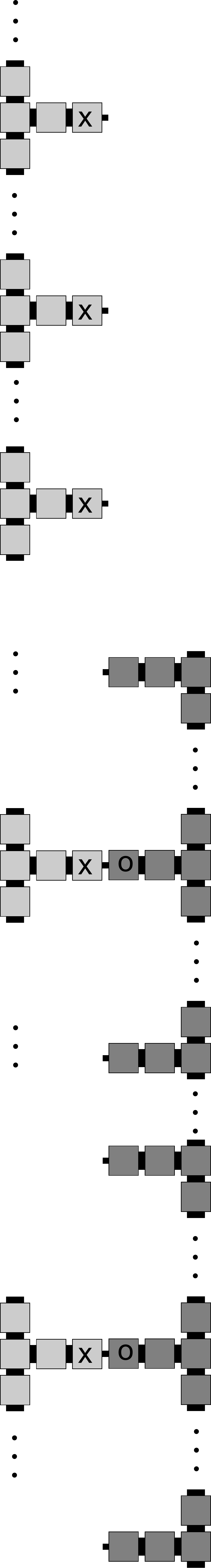}
\caption{The supertiles $\tilde{r_{bar}}$ and $\tilde{l}$ combining in $\prodasm{U}$.}
\label{fig:lrbar}
\end{center}
\end{figure}

However, $\tilde{R}(\tilde{r_{bar}}') = \tilde{{r_bar}}$ and $\tilde{R}(\tilde{l}') = \tilde{l}$, but $\tilde{R}(\tilde{\alpha_x}') \notin C^\tau_{\tilde{l}',\tilde{r_{bar}}'}$ because only $\tau-1$ rungs align in $\tilde{\alpha_x}'$ and therefore interact in $\mathcal{T}$ with strength at most $\tau-1$.  This is a contradiction.

} %
\fi

%% file: no-shift.tex
\ifabstract
\later{
\subsection{Dynamics of half-ladder assembly}

This section used the same notation as Section~\ref{sec:impossibility}.

For $q\in \N$ and each $1\leq j \leq q$, let $\tilde{r_{j}}$ be a copy of the $\tilde{r_p}$ supertile in $\prodasm{T}$, and for $k_j\in N$, let $\tilde{b_{k_j}}$ denote a copy of a bar in $\prodasm{T}$ consisting of $k_j$ tiles. Moreover, let $q$ be in $\N$, then for $1\leq j \leq q$, let $\tilde{r_{j}}'$ be a maximal simulator supertile in $\prodasm{U}$ that represents $\tilde{r_{j}}$, and let $\tilde{b_{k_j}}'$ be a maximal simulator supertile in $\prodasm{U}$ that represents $\tilde{b_{k_j}}$.

In addition, for $1\leq j \leq q-1$ and $k_j\in \N$ for each $j$, we let $\tilde{rb}$ be the supertile in $\prodasm{T}$ that is obtained by the assembly of supertiles that assemble via the sequence described as follows. For convenience, let $\tilde{rb}_{0}$ denote $\tilde{r}_1$. Then, we let $\tilde{rb}_j^*$ denote the supertile that results when $\tilde{rb}_{j-1}$ and $\tilde{b}_{k_j}$ bind via the $\tau$-strength glue labeled $8$ in Figure~\ref{fig:ladders-tile-set} exposed by the southeasternmost tile of $\tilde{rb}_{j-1}$ and the northernmost tile of $\tilde{b}_{k_j}$. Moreover, we let $\tilde{rb}_j$ denote the supertile obtained when $\tilde{rb}_*$ and $\tilde{r}_{j+1}$ bind via the $\tau$-strength glue labeled $9$ in Figure~\ref{fig:ladders-tile-set} exposed by the southernmost tile of $\tilde{rb}_*$ and the northeasternmost tile of $\tilde{r}_{j+1}$.  Then we let $\tilde{rb} = \tilde{rb}_{q-1}$. Note that $\tilde{rb}$ is a $\tau$ rung right half-ladder.

Then, since (1) $\mathcal{U}$ strongly simulates $\mathcal{T}$ and (2) for each $j$ such that $1\leq j\leq q$, $\tilde{r}_{j}'$ and $\tilde{b}_{k_j}$ are maximal simulator supertiles, it follows from Claim~\ref{ob:max} that the assembly sequence described above for $\tilde{rb}$ gives rise to an assembly sequence in $\mathcal{U}$ that corresponds to replacing each supertile in the assembly sequence described above with the appropriate $m$-block macrotile representative in $\prodasm{U}$. We let $\tilde{rb}'$ denote the supertile in $\prodasm{U}$ that results from this assembly process.

The following lemma puts a restriction on how $\tilde{rb}'$ can assemble.
Intuitively, Lemma~\ref{lem:noShift} states that as $\tilde{rb}'$ assembles, each consecutive supertile $\tilde{r}_{j+1}'$ that attaches in the assembly of $rb'$ does so such that it is ``aligned'' with the previous $\tilde{r}_{j}'$. Note that even though the lemma here is stated for right half-ladders an analogous lemma holds for left half-ladders.

\begin{lemma}\label{lem:noShift}
For each $j$ such that $1\leq j\leq q-1$ and vector $\vec{v} = (0,m)$ (where $m$ is the macrotile size), the subassemblies $r_{j}'$ and $r_{j+1}'$ of $rb' \in \tilde{rb}'$ satisfy the following equation: $r_{j}'(p - k_j \vec{v}) = r_{j+1}'(p)$ for all $p\in\Z^2$.
\end{lemma}

\begin{proof}
Notice that for a fixed value of $j$, there is \emph{some} vector $\vec{v} = (x,y)$ with $x,y\in \Z^2$ with $y > 0$ and some constant $c\in \N$ with $c>0$, such that $r_{j}'(p - c \vec{v}) = r_{j+1}'(p)$. For this fixed $j$, we first show that $x=0$ and $y=m$ by contradiction. First, for the sake of contradiction, suppose that, $x\neq 0$. Then, note that since $\tilde{r_{j}}'$ and $\tilde{b}_{k_j}'$ are maximal simulator supertiles, they can bind in a single assembly step to form a supertile representative of the assembly in $\prodasm{T}$ resulting when the supertiles $\tilde{r}_{j}$ and $\tilde{b}_{k_j}$ bind via the $\tau$-strength glue labeled $8$ in Figure~\ref{fig:ladders-tile-set} exposed by the southeasternmost tile of $\tilde{r}_{j}$ and the northernmost tile of $\tilde{b}_{k_j}$.  For the same reason,
$\tilde{r}_{ij}'$ and $\tilde{b}_{k_j}'$ can bind in a single assembly step to form a supertile representative of the supertile in $\prodasm{T}$ resulting when the supertiles $\tilde{r}_{ij}$ and $\tilde{b}_{k_j}$ bind via the $\tau$-strength glue labeled $9$ in Figure~\ref{fig:ladders-tile-set} exposed by the northeasternmost tile of $\tilde{r}_{j}$ and the southernmost tile of $\tilde{b}_{k_j}$. Moreover, for a constant $\eta$ in $\N$, $\eta$ copies of $r_{j}'$ and $\eta - 1$ copies of $b_{k_j}' \in \tilde{b_{k_j}}'$ can bind in $2\eta - 1$ assembly steps to form a supertile representative of the assembly, $\beta$ say, in $\prodasm{T}$ resulting when the $\eta$ copies of $\tilde{r}_{j}$ and $\eta - 1$ copies of $\tilde{b}_{k_j}$ bind via the $\tau$-strength glues labeled $8$ and $9$ in Figure~\ref{fig:ladders-tile-set} appropriately. Call this representative assembly in $\mathcal{U}$ $\beta'$, and for $1\leq z \leq \eta$, let $\rho_z$ denote the $z^{th}$ copy of $r_{j}'$ in $\beta'$ starting from the northernmost copy of $r_{j}'$ in $\beta'$.

Notice that for each $z< \eta$ and $p\in \Z^2$, $\rho_z(p + c\vec{v}) = \rho_{z+1}(p)$. In particular, under the assumption that $x\neq 0$, the $x$-values of the tile locations of $\rho_{z+1}$ are equal to the $x$-values of the tile locations of $\rho_z$ after shifting by $c*x$. Therefore, for $\eta > 3m$, the $x$-values of the tile locations of $\rho_{z+1}$ are equal to the $x$-values of the tile locations of $\rho_z$ after shifting by more than $3m*c*x$. This is a contradiction, since these shifted tiles contain tiles at locations outside of an $m$-block region that either maps to a tile or is part fuzz for our $m$-block representation of $\beta$, which violates the definition of simulation. See Figure~\ref{fig:shifted-x} for an example.

\begin{figure}[htp]
\centering
  \subfloat[][]{%
        \label{fig:shifted-rung}%
        \makebox[2in][c]{ \includegraphics[scale=1]{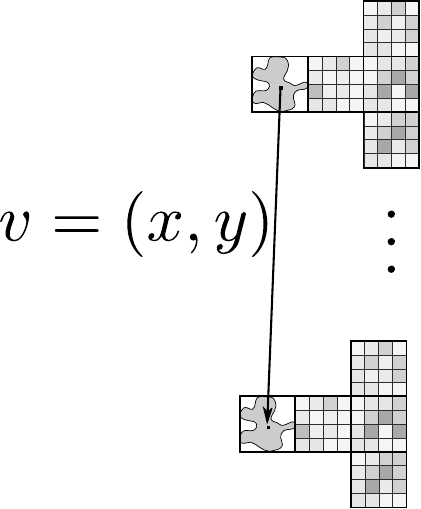}}
        }%
        \quad
  \subfloat[][]{%
        \label{fig:shifted-x-sub}%
        \makebox[2in][c]{ \includegraphics[scale=.4]{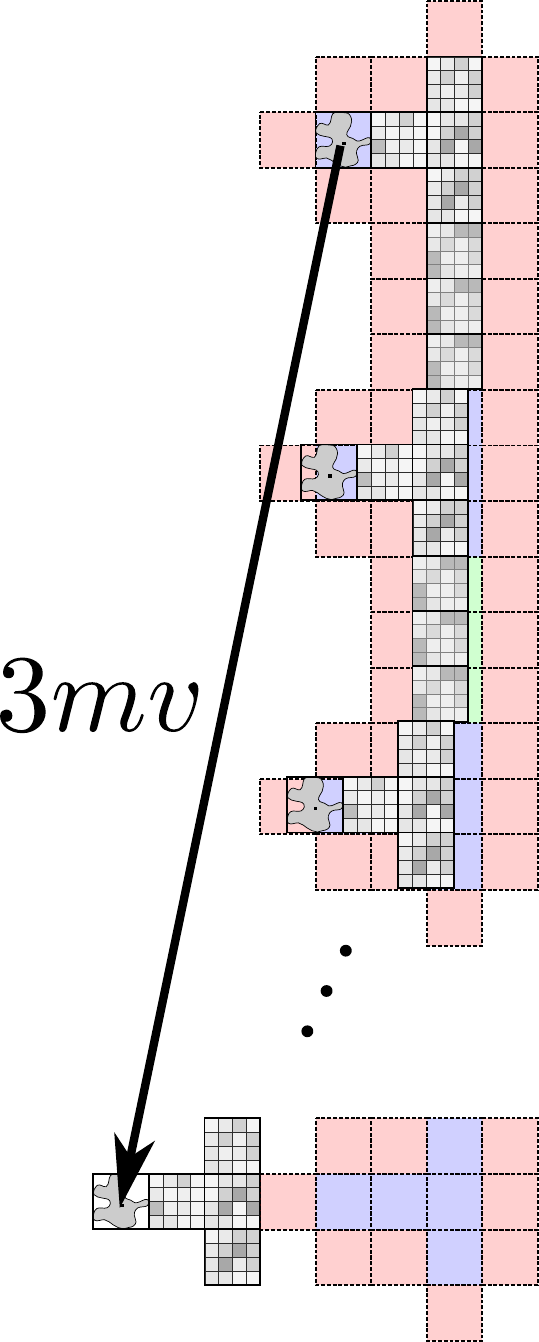}}
        }%
  \caption{(a) Two single rungs of distinct $r_{j}'$ subassemblies of $rb$ are shown. In general $r_{j}'$ will have $\tau$ rungs.  (b) A depiction of $\rho_\eta$ where $\vec{v}=(x,y)$, $x\neq0$, and $\eta= 3m\vec{v}$.}
  \label{fig:shifted-x}
\end{figure}

Therefore, the vector $\vec{v}$ must be of the form $(0,y)$. To complete the proof, it suffices to show that $-c*y = -k_j*m$. Again, we prove this by contradiction. Therefore, for the sake of contradiction, suppose that  $-c*y \neq -k_j*m$. First, if $-c*y > -k_j*m$, then for $\eta \geq 3m$, the $m$-block region of $\beta'$ representing the southernmost tile of $\beta$ must be empty. See Figure~\ref{fig:shifted-y2} for an example in this case. This is a contradiction since the representation function cannot map any empty $m$-block region to a tile. Finally, if $-c*y < -k_j*m$, then for $\eta \geq 3m$, $\rho_\eta$ must contain a tile at a location outside of an $m$-block region that either maps to a tile or is part fuzz for our $m$-block supertile representation of $\beta$, which once again violates the definition of simulation. See Figure~\ref{fig:shifted-y2} for an example in this case. Again, we arrive at a contradiction. Therefore, we see that $-c*y = - k_j*m$. Hence, we may choose the vector $\vec{v}$ to be $(0,m)$ and the constant $c$ to be $k_j$.

\begin{figure}[htp]
\centering
  \subfloat[][]{%
        \label{fig:shifted-y2}%
        \makebox[2in][c]{ \includegraphics[scale=.4]{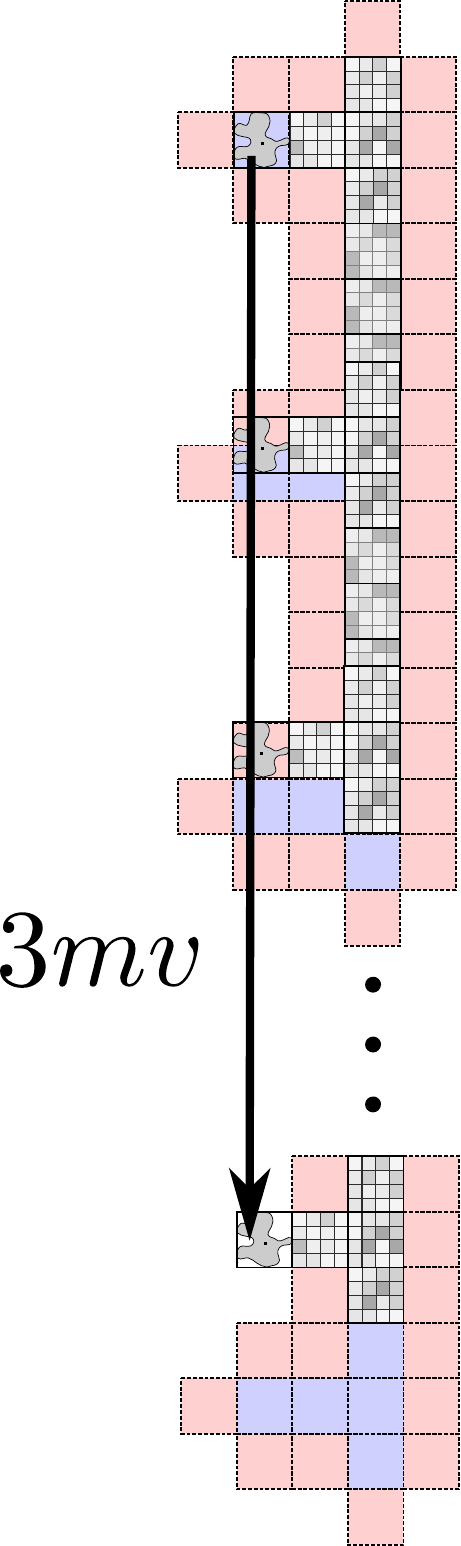}}
        }%
        \quad
  \subfloat[][]{%
        \label{fig:shifted-y1}%
        \makebox[2in][c]{ \includegraphics[scale=.4]{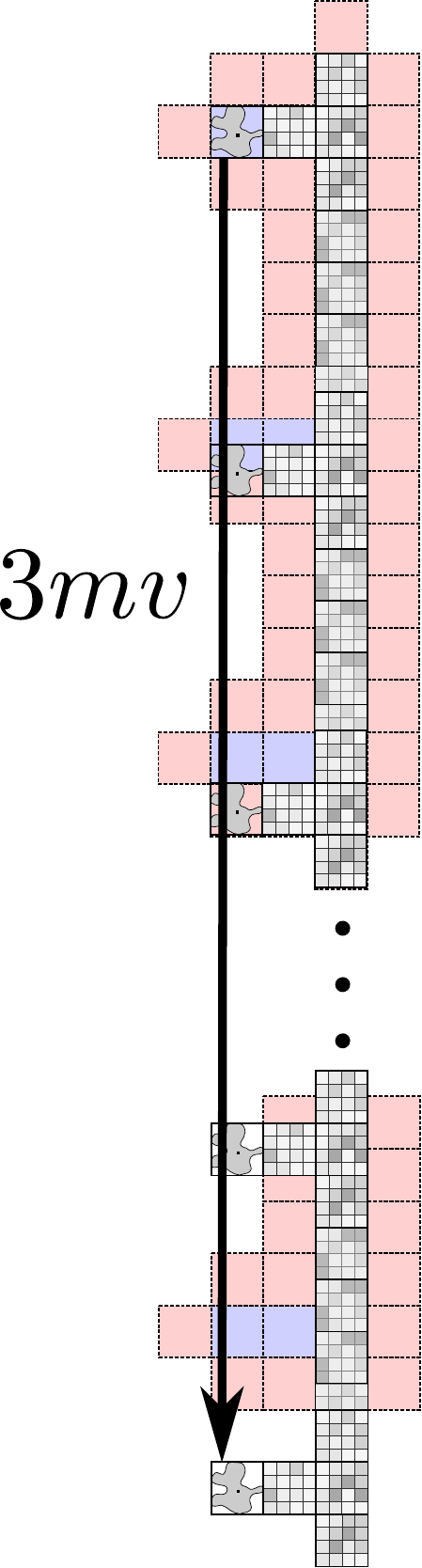}}
        }%
  \caption{(a) A depiction of $\rho_\eta$ where $\vec{v}=(0,y)$, $-c*y > -k_j*m$, and $\eta= 3m\vec{v}$. (b) A depiction of $\rho_\eta$ where $\vec{v}=(0,y)$, $-c*y < -k_j*m$, and $\eta= 3m\vec{v}$. Note that in this figure, as in Figure~\ref{fig:shifted-rung}, only single rungs are depicted. In general, each $r_{ij}'$ subassembly of $\rho_\eta$ will have $\tau$ rungs.}
  \label{fig:shifted-y}
\end{figure}

\end{proof}
} %
\fi

%% file: SimLadders2.tex
\section{Simulating Arbitrary Lower Temperature Ladder Systems}\label{sec:simLadder}
\vspace{-10pt}

We now prove that, even though higher temperature systems can only strongly simulate lower temperature ladder systems if a uniform mapping exists between the temperatures, a uniform mapping is not required for (standard) simulation.

\ifabstract
\later{
\section{Proof from Section~\ref{sec:simLadder}: Simulating Arbitrary Lower Temperature Ladder Systems}

}%
\fi

\both{
\begin{theorem} \label{thm:simLadder2}
For $\tau, \tau' \in \mathbb{N}$ where $1<\tau<\tau'$, let $\mathcal{T}$ be the ladder system at temperature $\tau$. Then, there exists a system $\mathcal{S}$ at temperature $\tau'$ which simulates $\mathcal{T}$.
\end{theorem}
}%

At a high-level, the construction which proves Theorem~\ref{thm:simLadder2} works by leveraging nondeterminism and the fact that for each pair of supertiles $\ta,\tb \in \prodasm{\calT}$ which are able to $\tau$-stably combine, for each $\ta' \in \prodasm{\calS}$ where $\tilde{R}(\ta') = \ta$, there simply must exist some $\tb' \in \prodasm{\calS}$ where $\tilde{R}(\tb') = \tb$ and $\ta'$ and $\tb'$ can $\tau'$-stably combine, but there may be many other $\tb'' \in \prodasm{\calS}$ where $\tilde{R}(\tb'') = \tb$ such that $\ta'$ and $\tb''$ cannot $\tau'$-stably combine.  Specifically, for each side of half-ladder, there are multiple types which can form, each with exactly $0$ or $1$ ``special'' rungs.  (See Figure~\ref{fig:ladders-standard-simulation} for a schematic example.)  All rungs on a left half-ladder can combine with all rungs on a right half-ladder with strength $1$, but whenever rungs of the same type combine, they do so with strength $\tau' - \tau + 1$.  The formation of all half-ladder supertiles guarantees that any pair of oppositely facing half-ladders can have no more than one pair of rungs with matching types, and for each half-ladder with $\tau$ or more rungs there exists a producible oppositely facing half-ladder with rungs in matching locations and one of them matching in type.  (Note that $\calS$ simulates at scale factor $2$.)  In such a way, $\tau$ rungs in matching locations of two oppositely facing half-ladders all guaranteed to be sufficient and necessary to form a ladder, and all possible half-ladder and ladder representing supertiles are producible, making $\calS$ correctly simulate $\calT$.

\ifabstract
\later{
\begin{proof} (Proof of Theorem~\ref{thm:simLadder2})
To prove Theorem~\ref{thm:simLadder2}, let $\tau,\tau' \in \N$ be arbitrary temperatures such that $1 < \tau < \tau'$, let $\mathcal{T} = (T,\tau)$ be the ladder system at temperature $\tau$ (see Figure~\ref{fig:ladders-tile-set} for the tile set $T$), let $\mathcal{S} = (T',\tau')$ be the system which simulates it at temperature $\tau'$, and let $R: B^{S}_2 \rightarrow T$ be the representation function mapping blocks of tiles from $T'$ to tiles of $T$. We will now show how to construct $T'$ such that $\mathcal{S}$ simulates $\calT$.  First, we will note that the scale factor of the simulation will be $2$, i.e. each tile of $T$ will be represented by a $2 \times 2$ block of tiles from $T'$.  Then, for each tile type $t \in T$, we will create $4$ tile types for $T'$ so that they can form a $2 \times 2$ square, and design each pair of matching glues within the interior of each block to be unique in all of $T'$ and with strengths set as follows. Each east or west glue will be of strength $\tau'$, each western pair of north and south glues will be of strength $\lceil \tau'/2 \rceil$, and the eastern pair of north south glues will be of strength $\lfloor \tau'/2 \rfloor$.   (See Figure~\ref{fig:2x2-block} for an example.) Any glue which was on the exterior of $t$ is now represented on the corresponding side of the $2 \times 2$ block, but instead split into two glues, one whose strength is $\lceil \tau'/2 \rceil$ and one $\lfloor \tau'/2 \rfloor$ (with the convention that the stronger is on the left or top).  The one exception is the glues at the end of rungs, which remain as single strength-$1$ glues on the bottom tile at the end of each rung for now.  Note that because of the strengths of the glues on their interiors, each block must form by the top two tiles and bottom two tiles each first combining in pairs, and then those two pairs can combine to form the full block.  We will now describe the further modifications to $T'$.

\begin{figure}[htp]
\centering
  \subfloat[][A tile/block not at the end of a rung]{%
        \label{fig:2x2-block}%
        \includegraphics[width=2.5in]{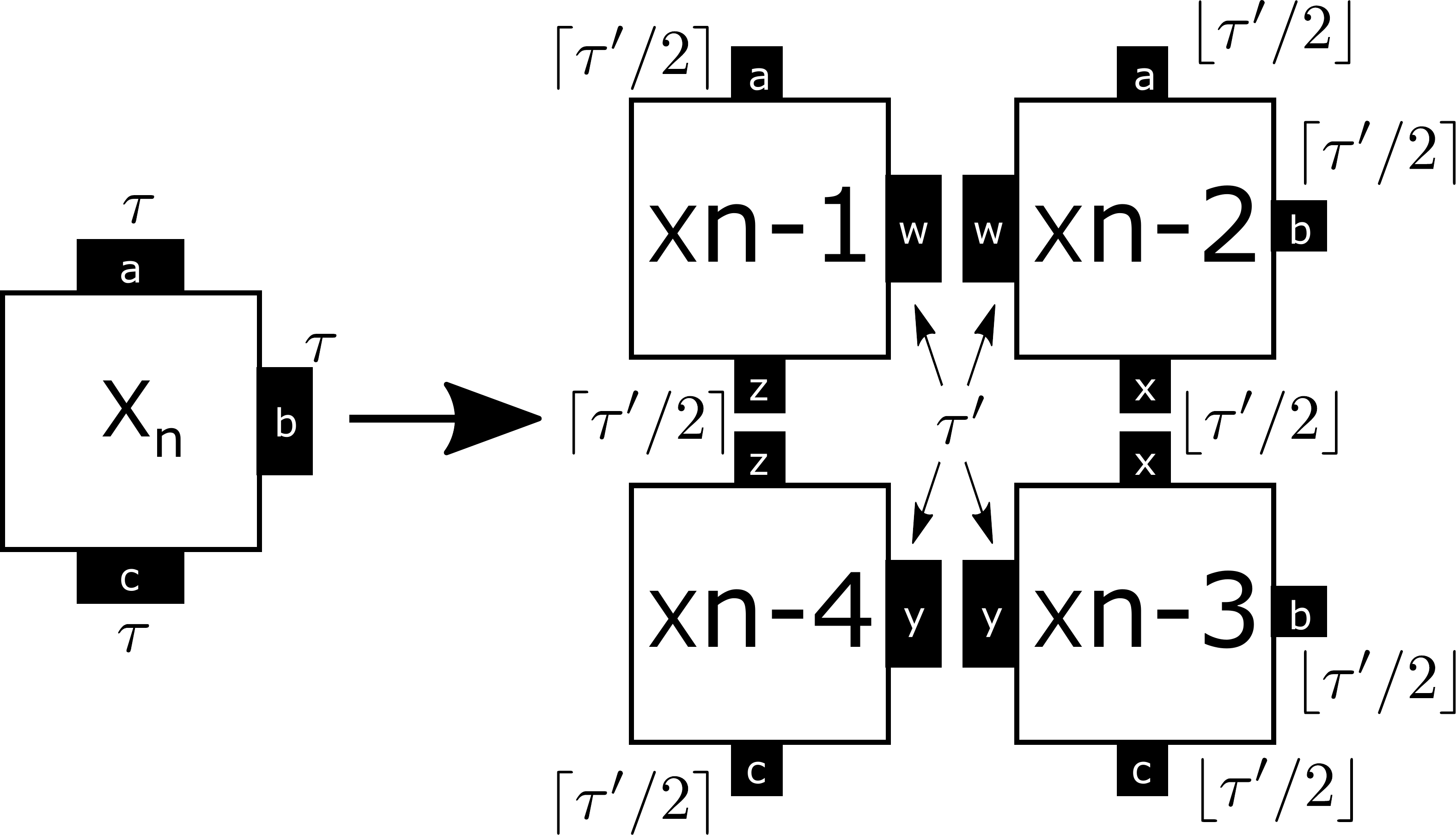}}
        \quad
  \subfloat[][A tile/block at the end of a rung]{%
        \label{fig:2x2-block-rung}%
        \includegraphics[width=2.5in]{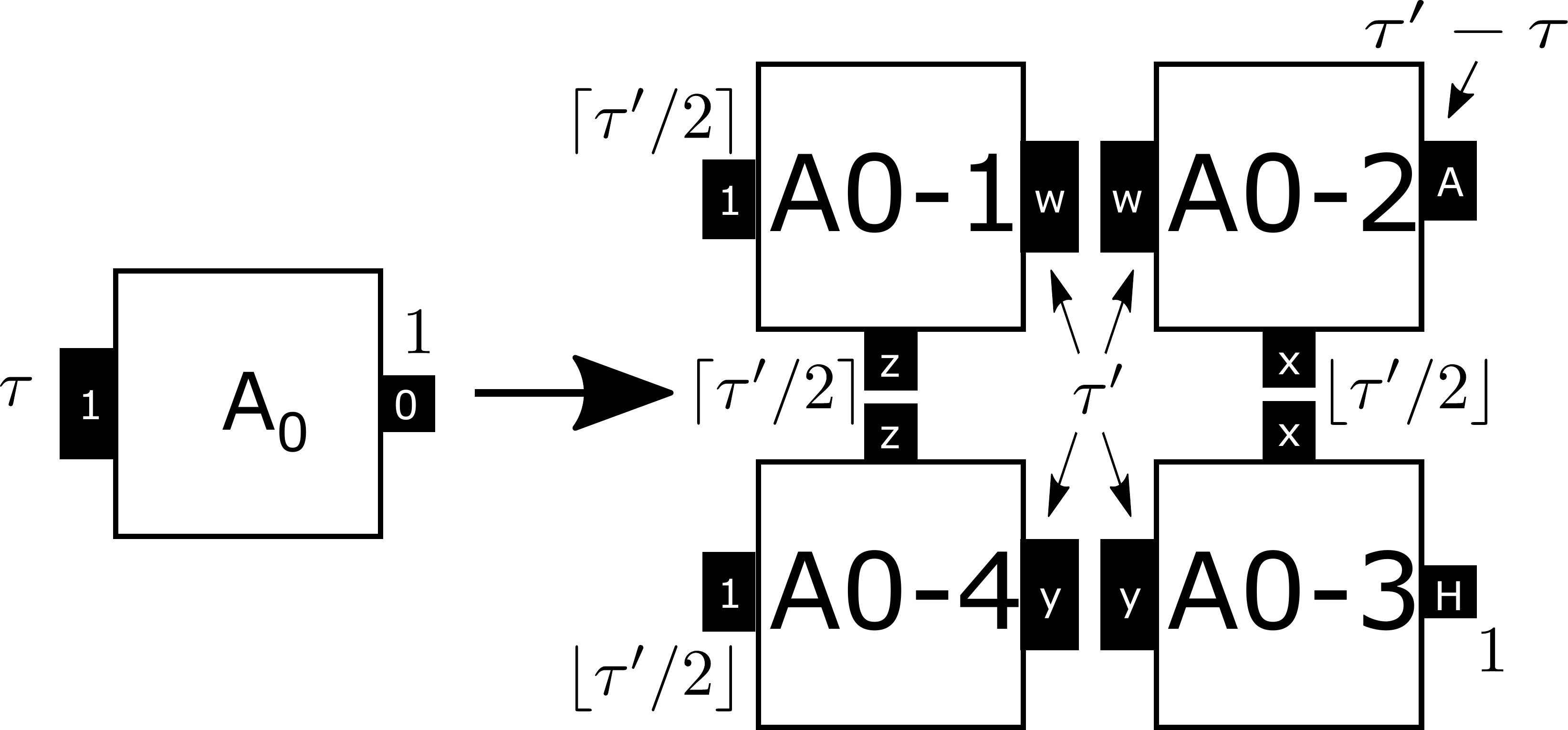}}
  \caption{Example creation of the initial $2 \times 2$ blocks of tiles in $T'$ (right of each) from a single tile in $T$ (left of each).  Glue strengths are labeled outside of glue boxes. Note that the $w$, $x$, $y$, and $z$ glues generically represent glues which are unique to each block and also to each location within that block.}
  \label{fig:no-strict}
\end{figure}

For both the left half-ladder tiles and right half-ladder tiles created for $T'$, we will make $2$ subsets of tiles by creating $2$ distinct copies of each of those tiles, for a total of $4$ subsets.  We will call the two subsets of tiles for left half-ladders the $B$ and $C$ sets, and those for the right half-ladders the $A$ and $D$ sets.  Each set will be designed so that only tiles from the same set can combine with each other within a half-ladder (i.e. no $A$ tile can bind to a $D$ tile, and no $B$ tile can bind to a $C$ tile, in the same half-ladder) by making each to have distinct glues unique to the type of that half-ladder at all locations other than those at the ends of rungs.  Now we will modify the $2$ tiles which form the ends of each rung of each subset of types.  As shown in Figure~\ref{fig:2x2-block-rung}, we make the bottom tile of the end of each rung expose a strength $1$ glue of type $H$.  Then, for each half-ladder type (i.e. $A$, $B$, $C$, or $D$) we make the top tile at the end of each rung expose a strength $\tau' - \tau$ glue of type matching the half-ladder type ($A$, $B$, $C$, or $D$).  We will thus call the rung type the same as the label of that top glue.

The next modification for each subset will be to create $4$ new tiles for another $2\times 2$ block which will be at the base of a rung, i.e. as a leftmost (rightmost) block of a left (right) half-ladder which attaches to blocks above, below, and to the right (left) to initiate a rung.  The tiles of each such block will be formed similarly to the other tiles of its group, except that for this block, the type of rung which can attach will not match the type of the half-ladder.  We will call this a \emph{special} rung for the half-ladder, and the following listing shows the pairs consisting of first the type of half-ladder, and then second the type of the special rung: $\{(A,C),(B,A),(C,D),(D,B)\}$.  The glues on the side of this $2\times 2$ block facing the rung are unique to the type of rung to attach, and two new blocks are created for each rung so that the rung has the correct type by ending with an exposed $H$ glue on the bottom and a glue of the matching type on the top.  Figure~\ref{fig:ladders-standard-simulation} shows how a half ladder of each type could contain rungs solely of its type or also include one of the special type designated for it.

The final modification to each of the subsets of half-ladder tiles consists of once again making a duplicated copy of subsets of tiles.  This time, for each type of half-ladder, we make two copies of each of the tile types which form the blocks of its ``backbone'' (i.e. the vertical column on the left side of a left-half ladder or the right side of a right half-ladder), except for the block which connects to the special rung.  For instance, for the $A$-type half-ladder tiles that form the backbone blocks (excluding the block attaching to the special rung), we make one copy which we will refer to as the \texttt{top}, and one we refer to as the \texttt{bottom}.  We augment each glue on the north and south sides of the tiles which are at the boundaries of those blocks by adding a $U$ to those of the top set and a $D$ to those of the bottom set.  This essentially ensures that there are two independent sets of backbone blocks for each type of half-ladder which are unable to attach to each other.  Finally, we mark the southern glues of the southern tiles of the backbone block which attaches to the special rung with a $D$ and the northern glues of the northern tiles of that block with a $U$.  (See Figure~\ref{fig:ladder-tile-sets-A-and-B} for an example of the tile types for the $A$ and $B$ type half-ladders.  Special rungs are located in the middle of each example half-ladder formation.)

The block representation function $R$ maps all blocks of tiles of $T'$, or portions of blocks containing the top $2$ tiles, to the tiles in $T$ which can appear in corresponding locations of (half-)ladders in $\calT$.  This means that all blocks, or top halves of blocks, of each type of half-ladder, irrespective of being in the \texttt{top} or \texttt{bottom} sets or part of special rungs, map to the base tiles of $T$ from which each originated as a copy (possibly after multiple steps of copying).  For instance, every block in $\mathcal{S}$ which is in a backbone of a left-ladder and also connects to a rung (or can connect to a rung), maps to tile type $A_2$ of $T$ (see Figure~\ref{fig:ladders-tile-set}).  This includes blocks of $B$ and $C$ types, those which connect to $B$ and $C$ type rungs, respectively, in both the \texttt{top} and \texttt{bottom} sets of each type, as well as to special rungs of types $A$ and $D$, respectively.

}%
\fi

\begin{figure}[htp]
\begin{center}
\includegraphics[width=3.0in]{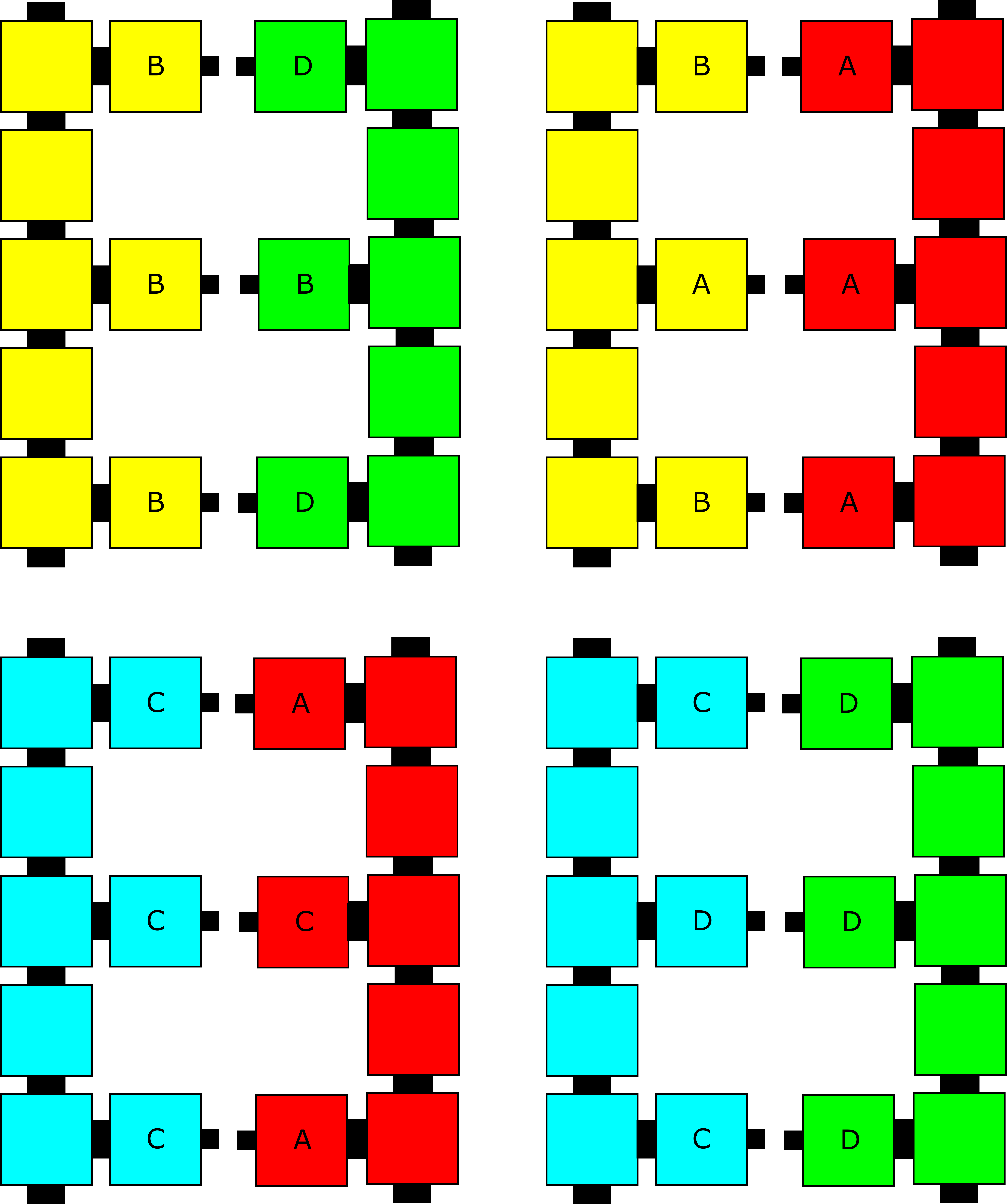}
\caption{Intuitive sketch of the set of half-ladders possible in the high temperature system $\calS$ which simulates a low temperature ladder system $\calT$, shown without scaling.  Yellow: $B$-type half-ladders, Blue: $C$-type half-ladders, Red: $A$-type half-ladders, Green: $D$-type half-ladders. Each type of half-ladder is shown once with no special rung and once with one special rung (the most possible), and each is paired with the type of half-ladder with which it could bind if each had at least $\tau$ rungs in matching locations (after translating appropriately).  Note that the spacing and ordering of rungs can be arbitrary, and also that spacing tiles are left out for compactness, so rungs are closer together and shorter than they would actually be.  All pairs of rungs of different types bind with each other with strength $1$ (due to the $H$ glues on their bottom tiles - not shown), and all pairs of rungs of the same type bind with strength $\tau' - \tau + 1$ due to the sum of the $H$ glue (strength 1) and ``type'' glue (strength $\tau' - \tau$) bindings.}
\label{fig:ladders-standard-simulation}
\end{center}
\end{figure}

\ifabstract
\later{
We now argue why the tile set $T'$, thus constructed, when used in system $\mathcal{S}$ at temperature $\tau'$, correctly simulates $\calT$.  The first key observation is that, in any producible assembly $\alpha' \in \prodasm{\mathcal{S}}$ which represents some portion of a half-ladder, i.e. $\tilde{R}(\alpha') = \alpha$ where $\alpha \in \prodasm{\calT}$ and $\alpha$ is a supertile containing only tiles from either a left half-ladder or a right half-ladder, $\alpha'$ can contain either $0$ or $1$ special rungs. Since the tile types for each half-ladder type are constructed similarly, with the main variance being the type of special rung associated with each half-ladder type, without loss of generality we will discuss the formation of left half-ladders of type $B$.  As shown in Figure~\ref{fig:ladder-B-tile-set}, the tiles for $B$-type ladders can form blocks, or the top or bottom halves of blocks, independently, but no single tile of one block can combine with any portion of a different block because all glues external to all blocks are of strength $< \tau'$.  Portions of blocks must be at least size $2$ before they can interact with other (possibly partial) blocks on their north or south, and must be fully formed before they can interact with blocks on their east or west.  Because half-blocks can interact with other blocks and half-blocks, we will now analyze how they are handled.

The representation function $R$ is defined so that it maps any complete block or any top pair of tiles from a block over $T'$ to the corresponding tile in $T$.  That means that any supertile in $\calS$ which contains just the top half of a block maps to a supertile in $\calT$ which has a tile in the corresponding location, but a supertile in $\calS$ which contains just the bottom half of a block maps to a supertile in $\calT$ which does not have a tile in the corresponding location.  Note that only supertiles representing half-ladders of the same side can potentially bind with each other along interior boundaries of blocks due to the fact that the only glues that left half-ladder tiles and right half-ladder tiles have in common are at the ends of rungs, so without loss of generality, we'll focus on left half-ladder supertiles.  Whenever two left half-ladder supertiles in $\calS$ combine using the exterior glues of blocks, then the mapping of the supertiles in $\calS$ to those in $\calT$ is straightforward and doesn't require additional discussion at this point.  However, when two supertiles combine by using the interior glues of a block, then the first thing to note is that those must be the north/south glues of the interior of the block because the only way any portion of a block can be attached to a block on its east or west is for it to be a complete block since to bind with strength $\tau'$ on those sides requires both the bottom and top glues of that side which can only be present if both halves of the block have combined.  Then, we'll refer to the top supertile as $\ta'$ and the bottom as $\tb'$ and note that $\ta'$ maps to a supertile in $\calT$ with a tile in the location corresponding to the partial block, but $\tb'$ maps to a supertile in $\calT$ with that position empty, so the definition of simulation holds.  The additional remaining cases occur when either $\ta'$ or $\tb'$ has an entire block at the location where they could potentially combine and the other has half of a block.  If both have the top halves of their blocks, then they each map to supertiles in $\calT$ which have tiles at that location, so their inability to bind in $\calS$ again matches the fact that they wouldn't be able to bind in $\calT$.  However, if $\ta'$ has the entire block and $\tb'$ has only the bottom half of the block, $\ta'$ would map to a supertile in $\calT$ which has a tile in that location but $\tb'$ would map to one which does not, and thus the two supertiles in $\calT$ that they map to would be able to combine, but $\ta'$ and $\tb'$ would not.  However, by the definition of simulation, $\calS$ must weakly model $\calT$, and by this definition it must simply be the case that there exists some supertile which maps to the same thing as $\tb'$ to which $\ta'$ can bind (and vice versa), and this clearly exists since the supertile which consists of $\tb'$ minus the bottom half of the block is clearly producible because that portion of the block must not be attached to anything other than the block to its south, and the two tiles to which it is bound are stably bound to each other regardless of its attachment (via their shared $\tau'$-strength bond), and thus whatever assembly sequence produced $\tb'$ could be altered to produce $\tb''$ which is exactly $\tb'$ without this northmost half block by simply omitting the step where it binds.  Similarly, the $\ta''$ necessary to bind with $\tb'$ is also producible (i.e. it is the same as $\ta'$ but without the bottom half of its souther block).  Therefore, $\calS$'s simulation of $\calT$ is preserved regardless of the possible assembly sequences.

Because of the above argument, from this point we will only talk about supertiles in $\calS$ which are composed of completed blocks.

While forming half-ladders (i.e. before they are parts of half-ladders which combine to form ladders), blocks of a given type of half-ladder which are part of the \texttt{top} set can combine only with other blocks of that type and from the \texttt{top} set, or the block of that type which attaches to the special rung for that type.  (This holds similarly for those of the \texttt{bottom} set.)  We will now discuss those of the \texttt{top} set and note that the same holds for the \texttt{bottom} set.  The blocks of the \texttt{top} set of a $B$-type left half-ladder can combine to form arbitrarily tall backbones and with an arbitrary number of rungs of type $B$ with arbitrary spacing, which obviously models what the left half-ladder types of $\calT$ can form.  Additionally, a left half-ladder of type $B$ with $0$ or more type-$B$ rungs may combine with a type-$B$ backbone block which can connect to a special rung (which would be of type $A$).  Note that because this backbone block has \texttt{top} labeled glues only on its north side, such a block can only attach to the south of \texttt{top} left half-ladder.  However, \texttt{bottom} left half-ladders can form just as the \texttt{top}, and any can attach to the south side of the backbone block which attaches to the special rung.  In this way, left half-ladders of type $B$ can form with any possible number and spacing of rungs in the three following patterns (from bottom to top, and letting $B_{\texttt{bottom}}$ ($B_{\texttt{top}}$) denote a portion of a $B$-type half-ladder consisting only of blocks of the \texttt{bottom} (\texttt{top}) set): $B_{\texttt{bottom}}^*$, $B_{\texttt{bottom}}^*AB_{\texttt{top}}^*$, or $B_{\texttt{top}}^*$, which is essentially any pattern of type-$B$ rungs with either $0$ type-$A$ rungs, or exactly $1$ type-$A$ rung located at any particular rung location.  Furthermore, for each such half-ladder type, the analogous situation holds, with each half-ladder type being able to form with $0$ or $1$ special rungs and the other rungs in arbitrary number and at arbitrary spacing.  Since all left half-ladder blocks in $\calS$ map to left half-ladder tiles in $\calT$, regardless of their type, and vice versa for right-half ladders, $\calS$ correctly simulates $\calT$ in terms of half-ladder production and dynamics.

Now, we analyze the abilities of each possible type of half-ladder to combine with other half-ladders to form ladders. The key feature to note about the definition of simulation is its requirement that $\calS$ weakly models $\calT$.  Intuitively, what this requires is that whenever the simulator $\calS$ produces a supertile $\ta'$ which maps to some supertile $\ta$ in $\calT$, if $\ta$ can combine with $\tb$ to form $\tg$ in $\calT$, then there exists some $\tb'$ in $\calS$ which maps to $\tb$ and which $\ta'$ (or something that $\ta'$ can grow into while still mapping to $\ta$) can combine with to form $\tg'$ in $\calS$, and $\tg'$ maps to $\tg$.  The point to note is that there simply must \emph{exist} some such $\tb'$, and it need not be the case that \emph{any} arbitrary $\tb'$ (or something that any arbitrary $\tb'$ can grow into) that maps to $\tb$ can combine with a given $\ta'$.  (Essentially, each $\ta'$ must have some mate $\tb'$, but there may be many that it can never combine with. This is the key distinction between weakly models and strongly models, and thus between simulation and strong simulation.)  Therefore, we must simply show that (1) for every producible half-ladder $\ta'$ in $\calS$ with $\tau$ or more rungs, given that it maps to $\ta$ in $\calT$, any oppositely facing half-ladder $\tb$ in $\calT$ which can combine with $\ta$ to make a ladder has a corresponding oppositely facing half-ladder $\tb'$ in $\calS$ which $\ta'$ can $\tau'$-stably bind with it to form a ladder, and (2) no producible half-ladder in $\calS$ with $< \tau$ rungs can combine with any producible oppositely facing half-ladder to form a ladder.  (See Figure~\ref{fig:ladders-standard-simulation} for a schematic depiction of the following argument.)  Situation (1) can be easily shown by discussing the case of $B$-type half-ladders and noting that the argument is analogous for all other types.  There exist two main scenarios in this case.  Either (a) the $B$-type left half-ladder has $0$ special rungs of type $A$, or (b) it has $1$.  In case (a) where it has $0$, then for any right half-ladder in $\calT$ such that the left-half ladder that this one maps to can bind with it, there exists a producible $D$-type right half-ladder which maps to that in $\calT$ and which has exactly one $B$-type rung in a position which will align with a rung in the left-half ladder.  Since the half-ladders map to half-ladders with $\tau$ or more matching rungs (at the same relative offsets in $\calT$ and $\calS$), that means that at least $\tau-1$ pairs of rungs come together such that the left is type $B$ and the right is type $D$, and bind via their matching strength-$1$ $H$ glues, for a sum of $\tau-1$ strength bonding.  Additionally, one of the matching pairs is of two $B$-type rungs which will combine via their strength-$1$ $H$ glues and their strength $(\tau' - \tau)$ $B$ glues, for an overall binding strength of $(\tau' - \tau) + 1 + (\tau - 1) = \tau'$, and thus they can $\tau'$-stably bind.  In case (b), the $B$-type left half-ladder has at least $\tau-1$ type-$B$ rungs and one type-$A$ rung, and similar to the last argument but symmetric for the case of the $A$-type right half-ladder, there will exist a type $A$ right half-ladder which can combine along $\tau-1$ pairs of rungs (of types $B$ on the left and $A$ on the right) with strength $1$ and one pair of type-$A$ rungs for a total binding strength of $\tau'$.

The final thing to show is that ladders with fewer than $\tau$ matching rungs cannot form in $\calS$, since they cannot form in $\calT$. This follows directly from the observation that no half-ladder can have more than one special rung, and that for each possible half-ladder, whether or not it has a special rung, there exists no oppositely facing half-ladder with more than one rung which matches the types of any of its rungs.  (This can be seen in Figure~\ref{fig:ladders-standard-simulation}, since any additional rungs which could be on any half-ladder of a given type can only be of that type.) Therefore, any half-ladder which has $\tau'' < \tau$ rungs can find an oppositely facing half-ladder which has at most $\tau''$ rungs which can match with it, and at most one of the matching pairs can be of the same type.  Therefore, the maximum amount of binding strength between two half-ladders with $\tau''$ matching rungs is $\tau'' - 1$ from the pairs of rungs of different types and $(\tau' - \tau) + 1$ for the sum of the two glues on the matching pair of rungs, for a total of $\tau'' -1 + \tau' - \tau + 1 = \tau'' + \tau' - \tau < \tau + \tau' - \tau = \tau'$.  Thus, such a pair of half-ladders can bind with strictly less than $\tau'$ strength, meaning they cannot $\tau'$-stably combine.

We have thus shown that for any pair of producible assemblies $\ta',\tb' \in \prodasm{S}$ such that $\ta \rightarrow_{\mathcal{S}}^1 \tb$, $\tilde{R}(\ta) \rightarrow_\mathcal{T}^{\leq 1} \tilde{R}\left(\tb\right)$, since supertiles in $\calS$ can only either combine along block boundaries or the horizontal centers of blocks, and either way their combinations are followed by the assemblies in $\calT$ to which they map. Therefore, $\calT$ follows $\calS$.  We have further shown that whenever there exists a set of supertiles $\ta,\tb,\tg \in \prodasm{\calT}$ such that  $\ta$ can $\tau$-stably combine with $\tb$ to form $\tg$, then for all $\ta' \in \prodasm{\calS}$ such that $\tilde{R}(\ta') = \ta$, there exists some $\tb' \in \prodasm{\calS}$ where $\tilde{R}(\tb') = \tb$ such that either $\ta'$ or something that $\ta'$ can grow into while still representing $\ta$ can combine with $\tb'$ to form a $\tg' \in \prodasm{\calS}$ such that $\tilde{R}(\tg') = \tg$.  Therefore, $\calS$ weakly models $\calT$.  This proves that $\calS$ simulates $\calT$ at scale factor $2$ under the block representation function $R$ for arbitrary $1 < \tau < \tau'$.

\begin{figure}[htp]
\centering
  \subfloat[][Tile set for `$B$'-type left-half-ladders]{%
        \label{fig:ladder-B-tile-set}%
        \includegraphics[height=6.0in]{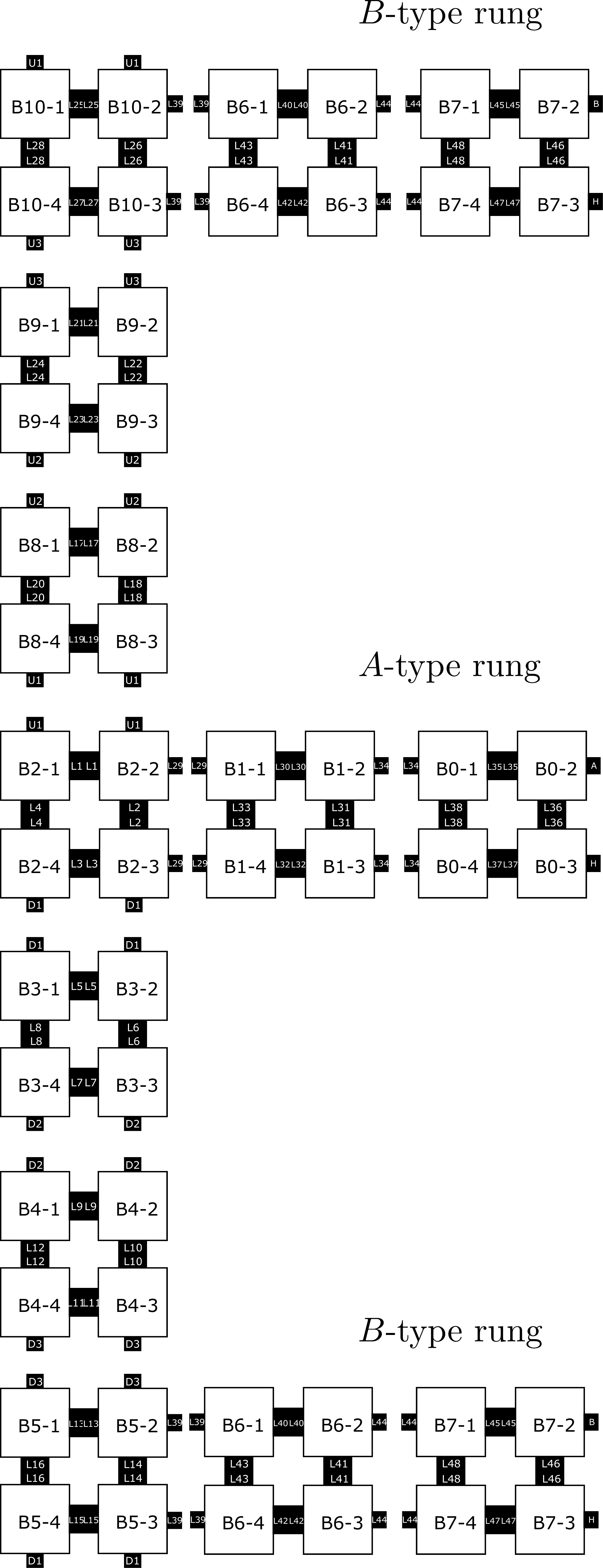}}
        \quad
  \subfloat[][Tile set for `$A$'-type right-half-ladders]{%
        \label{fig:ladder-A-tile-set}%
        \includegraphics[height=6.0in]{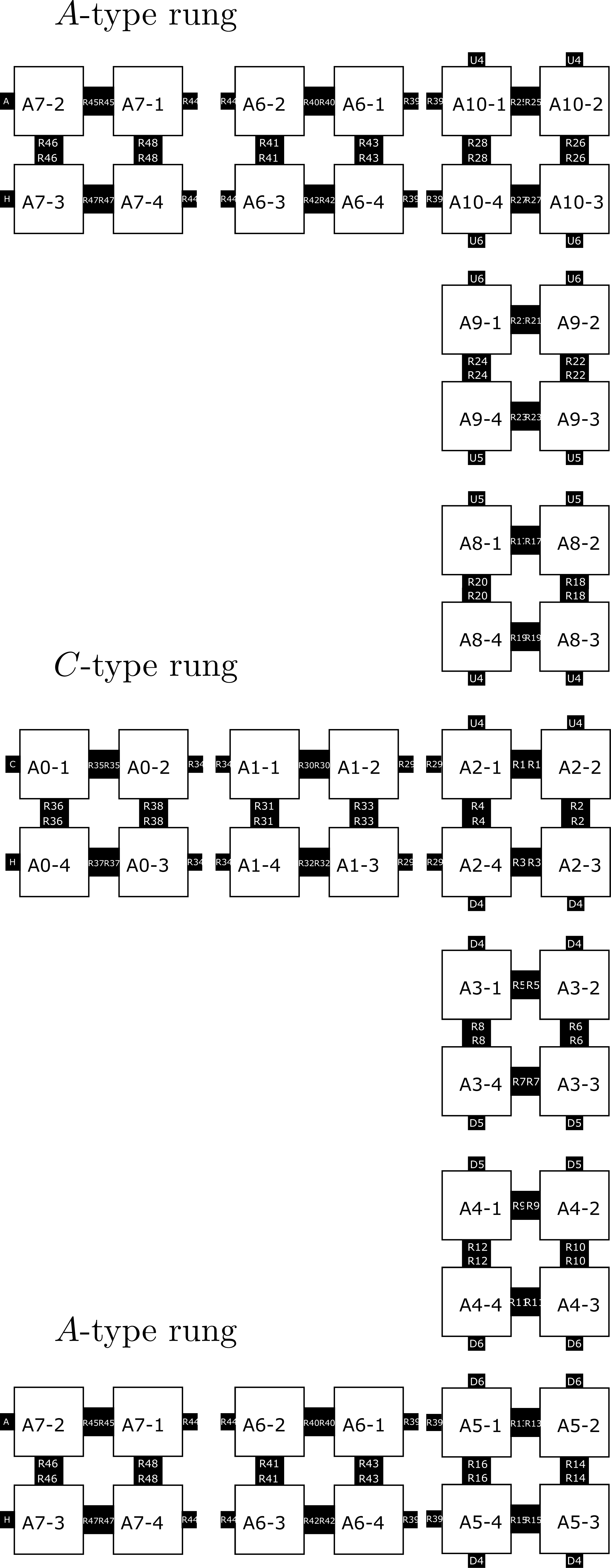}}
  \caption{Tile sets for high temperature simulation of ladders.  For each half-ladder, a special rung is shown as the middle rung.  Note that on each side, the top and bottom rungs contain duplicate tile types but are shown in each position to demonstrate how they can attach above or below special rungs.  All north/south glue pairs in the west halves of blocks (interior and exterior) are of strength $\lceil \tau'/2 \rceil$, and those in the east halves are of strength $\lfloor \tau'/2 \rfloor$.   For east/west glue pairs: those in the interior of blocks are of strength $\tau'$, those on the exterior of blocks (except for those at the ends of rungs) which are the north pair are of strength $\lceil \tau'/2 \rceil$ and the south pair are of strength $\lfloor \tau'/2 \rfloor$, and for those at the ends of rungs the top is strength $\tau'-\tau$ and the bottom is strength $1$.}
  \label{fig:ladder-tile-sets-A-and-B}
\end{figure}

\end{proof}
}%
\fi